\documentclass[11pt,a4paper]{article}

\usepackage{amsmath,amssymb,amsthm}
\usepackage[margin=25mm]{geometry}
\usepackage[numbers,sort&compress]{natbib}
\usepackage{bm}
\usepackage{graphicx}
\graphicspath{{./}{figures/}{upload/}}
\usepackage{booktabs}
\usepackage{url}
\usepackage{pgfplots}
\pgfplotsset{compat=1.16}

\newtheorem{prop}{Proposition}

\newtheorem{rem}{Remark}
\newtheorem{defn}{Definition}

\newcommand{\HH}{\mathcal{H}}
\newcommand{\Tr}{\operatorname{Tr}}
\newcommand{\ImPart}{\operatorname{Im}}
\newcommand{\RePart}{\operatorname{Re}}
\newcommand{\PV}{\operatorname{P}}
\newcommand{\dd}{\,d}
\newcommand{\bh}{\mathrm{BH}}
\newcommand{\inr}{\mathrm{int}}
\newcommand{\out}{\mathrm{out}}

\begin{document}

\title{Causality, Redshift, and Matching-Scale Cancellation for Reflective Black-Hole Horizons}

\author{Masahiro Kaminaga\\[0.5ex]
\small Faculty of Engineering, Tohoku Gakuin University, Sendai 984-8588, Japan}

\date{\empty}

\maketitle

\begin{abstract}
The Boltzmann flux reflectivity used in quantum-horizon models does not by itself define a causal black-hole boundary response.
We study two related questions for stationary Schwarzschild scattering.
First, a classical Hardy-space condition excludes an exact Boltzmann modulus on the entire positive frequency axis for any nonzero bounded causal reflection amplitude at a fixed matching plane.
We construct causal completions that preserve the modulus on any prescribed finite band.
Second, we prove a low-frequency outer-phase expansion for a class of completions with fixed profiles and a common Boltzmann slope, including an exact-band completion and a smooth crossover completion.
The logarithmic coefficient is common to this class, while the linear constant is an explicit functional of the completion.
When the fixed proper-frame cutoff is redshifted to the matching surface and the local temperature is the Tolman redshift of the Hawking temperature, the outer phase cancels the leading logarithmic dependence of the propagation phase on the matching distance.
The geometric remainder is quadratic in the proper distance, while a separate dispersive remainder depends on the completion.
This is a low-energy comparison of specified response families, not an exact matching invariance or a phase fixed by causality alone.
In an axial Schwarzschild example at a fixed matching plane,
accounting for the predicted linear-delay difference reduces the
maximum relative cavity-factor difference between two cutoffs
from about $21\%$ to $0.11\%$ on the sampled frequency band.
A passive quantum completion supplies the associated absorption-noise relation and checks the consistency of the asymptotic output.
\end{abstract}

\section{Introduction}
\label{sec:introduction}

Gravitational-wave echo models replace the purely ingoing black-hole boundary condition by a near-horizon reflection amplitude.
Repeated propagation between that region and the photon-sphere barrier produces an echo transfer function and a modified resonance spectrum \cite{Cardoso2016,Mark2017,CardosoPani2019}.
The response also depends on the source and on the separation of prompt and delayed propagation \cite{KaminagaEcho2026,ZhangRen2026}.
Effective boundary actions organize the low-energy boundary data \cite{Burgess2018}, while quantum-horizon and stretched-horizon models suggest possible reflection and dissipation mechanisms \cite{OshitaWangAfshordi2020,WangOshitaAfshordi2020,ManikandanRajeev2022,ChakrabortyEtAl2022,Terashima2025,Terashima2026}.

A central example is the Boltzmann flux reflectivity $\mathcal R_{\rm flux}=e^{-\beta_H|\widetilde\omega|}$ obtained by Oshita, Wang, and Afshordi \cite{OshitaWangAfshordi2020} and used in echo calculations \cite{WangOshitaAfshordi2020}.
Earlier horizon-continuation arguments also give a Boltzmann modulus \cite{KuchievFlambaum2005}, and the prescription appears in black mirror phenomenology \cite{AmaroSeoane2025}.
A modulus, however, does not specify a causal response between local scattering ports.
We ask whether the thermal law admits such a response globally, and which part of the resulting echo phase can be attributed to a matching position or an ultraviolet cutoff.

The first result is a global obstruction together with a finite-band construction.
Proposition~\ref{prop:boltzmann_obstruction} applies a classical Hardy-space theorem \cite{Garnett2007} to show that a nonzero $H^\infty(\mathbb C_+)$ reflection amplitude at a fixed plane cannot have the exact Boltzmann modulus at all positive frequencies.
This does not invalidate a low-frequency thermal prescription.
In fact, causal spectral factors can preserve that modulus exactly on any prescribed finite band.
We compare this construction with a smooth ultraviolet crossover.

The second result concerns a class of ultraviolet completions with the same low-frequency Boltzmann slope.
We select their outer factors, without an additional inner factor.
Proposition~\ref{prop:general_outer_phase} gives a common logarithmic coefficient, an explicit completion-dependent linear constant, and a controlled remainder under stated regularity and tail conditions.
The two constructions above are examples of this class.
If the response is specified locally, a fixed proper-frame scale $\Omega_{\rm UV}$ appears at radius $r_0$ as $\omega_c=\sqrt{f(r_0)}\Omega_{\rm UV}$; the analogous relation holds for the finite-band edge.
For the specified outer-response families with a fixed proper-frame
cutoff, the leading matching-distance term in the combined reflection
and propagation phase is
$(\beta/\pi-8M)\omega\log[d_0/(4M)]$.
With the local temperature related by Tolman redshift, this term
vanishes at the Hawking value $\beta=8\pi M$.
Proposition~\ref{prop:matching_scale_cancellation} and Eq.~(\ref{eq:band_matching_cancellation}) give the geometric and dispersive remainders.
The calculation compares redshifted response families.
It is distinct from exactly transporting the boundary data of one fixed physical system to a different reference plane.

The distinction from earlier results is important.
Burgess, Plestid, and Rummel derive the boundary-coupling flow that removes dependence on an arbitrary matching position \cite{Burgess2018}.
Oshita, Wang, and Afshordi obtain a logarithmic echo delay, which in Schwarzschild has the coefficient $8M$, in their dissipative model \cite[Eq.~(2.17)]{OshitaWangAfshordi2020}.
Our contribution is not the existence of either matching invariance or logarithmic delay.
It is the connection between the global causal restriction, explicit finite-band completions, and a phase expansion that separates the common logarithmic term from completion-dependent constants and remainders under explicit assumptions.
Causality permits additional inner factors and does not fix the full reflection phase.

We propagate the response through the axial Schwarzschild barrier and quantify how much of a fixed-plane cutoff change is carried by a linear delay.
A compact passive quantum completion then gives the thermal output through the same cavity.
This uses established input-output and quantum-scattering theory \cite{Beenakker1998,Caves1982,GardinerCollett1985,GardinerZoller}; the absorption-noise identities are consistency relations, not additional microscopic predictions of the horizon model.
Details on normalization, response kernels, more general inputs, stability, and an explicit unitary realization are collected in the appendices.

Section~\ref{sec:schwarzschild} specifies the fixed-plane response.
Section~\ref{sec:causal_completion} gives the causal obstruction and the two completions.
Section~\ref{sec:matching} derives their phases, the extension to a class of completions, and the matching-scale dependence.
Section~\ref{sec:composition} presents the Schwarzschild application and the passive thermal output.
Section~\ref{sec:conclusion} discusses the scope of the results.
We use $c=\hbar=G=k_B=1$ and the Fourier convention $e^{-i\omega t}$.

\section{A causal response at a fixed matching plane}
\label{sec:schwarzschild}

\subsection{Schwarzschild scattering and local ports}
\label{sec:channel_geometry}

For Schwarzschild mass $M$, define
\begin{equation}
f(r)=1-\frac{2M}{r},
\qquad
\frac{dr_*}{dr}=\frac{1}{f(r)}.
\label{eq:tortoise}
\end{equation}
Scalar, electromagnetic, and axial gravitational perturbations can be written in the form \cite{ReggeWheeler1957,Chandrasekhar1983}
\begin{equation}
\left[
\frac{d^2}{dr_*^2}
+\omega^2-V_{s\ell}(r)
\right]\psi_{s\ell\omega}=0,
\label{eq:rw_general}
\end{equation}
where
\begin{equation}
V_{s\ell}(r)
=f(r)
\left[
\frac{\ell(\ell+1)}{r^2}
+\frac{2M(1-s^2)}{r^3}
\right].
\label{eq:rw_potential}
\end{equation}
For $s=2$ this is the Regge--Wheeler potential for axial gravitational perturbations.

We use canonically normalized master modes; their action, commutators, and energy normalization are given in Appendix~\ref{app:normalization}.
Figure~\ref{fig:channel_schematic} defines the scattering ports.
At the matching plane, $a_-$ travels toward the unresolved near-horizon region and $a_+$ returns toward the barrier.
The modes $a_{\rm in}$ and $a_{\out}$ describe incidence from and propagation to infinity.
The extra port in the right panel supplies the passive quantum completion of Sec.~\ref{sec:quantum_completion}.

\begin{figure}[!htbp]
\centering
\includegraphics[width=\linewidth]{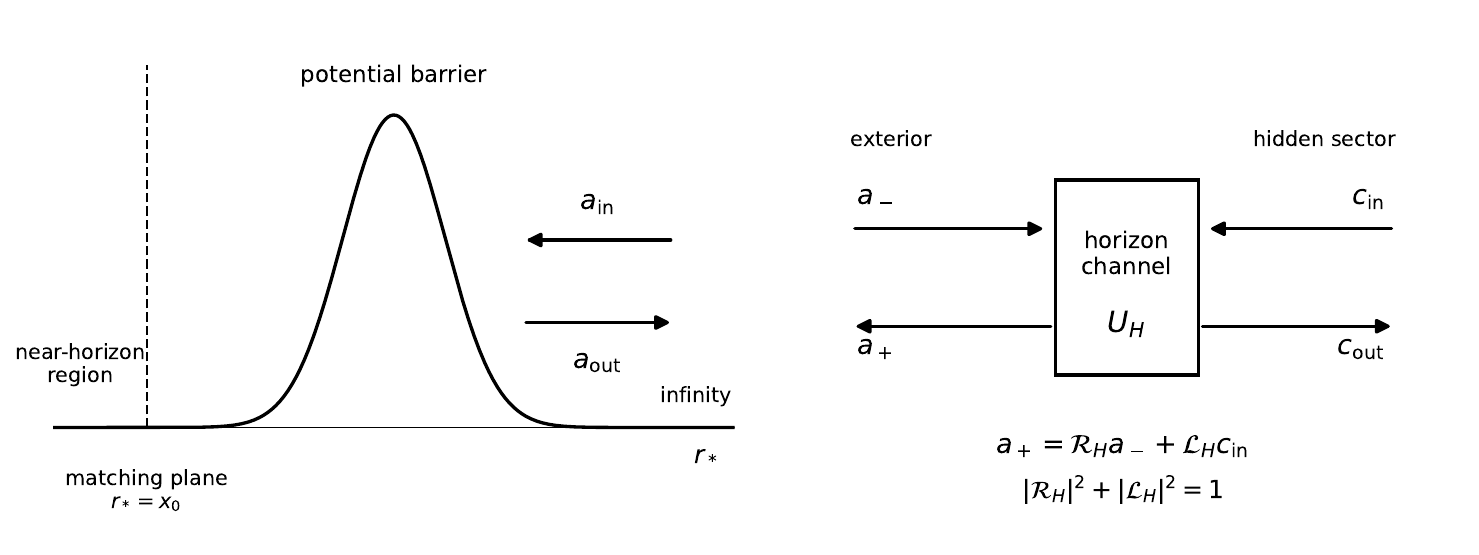}
\caption{Scattering geometry and a passive horizon-channel realization.
Left: the potential barrier, effective matching plane, and asymptotic modes; the potential profile is schematic.
Right: the joint unitary $U_H$ acts on the local exterior input $a_-$ and an effective environment input $c_{\rm in}$, with outputs $a_+$ and $c_{\rm out}$.
Tracing over the environment output gives the reduced exterior channel.
The displayed relation is its passive single-environment-port case.
The arrows in the right panel indicate channel flow, not the direction of the radial coordinate in the left panel.}
\label{fig:channel_schematic}
\end{figure}

We introduce a matching surface at
\begin{equation}
r_*=x_0,
\qquad
x_0\ll -M,
\label{eq:x0}
\end{equation}
where the Regge--Wheeler potential is exponentially small.
No material surface is assumed; $x_0$ is an effective matching scale.
To define an exact scattering problem at a finite plane, one can set the exterior potential to zero for $r_*<x_0$ and keep the Schwarzschild potential for $r_*\geq x_0$.
The plane-wave decomposition and the composition formulas are then exact for this matched model.
Replacing its amplitudes by the usual horizon-asymptotic Schwarzschild amplitudes is a near-horizon approximation, whose numerical size is checked in Appendix~\ref{app:numerical}.
On the free side of the matching plane,
\begin{equation}
\psi
=A_-e^{-i\omega(r_*-x_0)}
+A_+e^{i\omega(r_*-x_0)}.
\label{eq:near_horizon_modes}
\end{equation}
A frequency-dependent Robin condition
\begin{equation}
\psi'(x_0)=\Lambda(\omega)\psi(x_0)
\label{eq:admittance_bc}
\end{equation}
gives
\begin{equation}
{\cal R}_H(\omega)
:=\frac{A_+}{A_-}
=
\frac{i\omega+\Lambda(\omega)}
{i\omega-\Lambda(\omega)}.
\label{eq:RH_Lambda}
\end{equation}
The purely ingoing classical black-hole condition corresponds to
\begin{equation}
\Lambda_{\bh}(\omega)=-i\omega,
\label{eq:Lambda_BH}
\end{equation}
for which ${\cal R}_H=0$.

\subsection{Packet passivity and the analytic response}

For real $\omega>0$, the incoming and outgoing fluxes are proportional to $|A_-|^2$ and $|A_+|^2$.
Direct subtraction in Eq.~(\ref{eq:RH_Lambda}) gives
\begin{equation}
1-|{\cal R}_H(\omega)|^2
=-\frac{4\omega\ImPart\Lambda(\omega)}{|i\omega-\Lambda(\omega)|^2}.
\label{eq:passivity_identity}
\end{equation}
Thus $\ImPart\Lambda\leq0$ implies $|{\cal R}_H|\leq1$ on the positive real axis.
This real-frequency bound does not by itself establish causality or exclude an exponentially growing mode.

\begin{defn}[Admissible passive horizon response]
\label{def:admissible_response}
A scalar stationary horizon response is called admissible in this paper if
${\cal R}_H(z)$ is analytic and bounded for $\ImPart z>0$, has boundary
values satisfying $|{\cal R}_H(\omega)|\leq1$ for real $\omega$, and obeys
the reality condition
\begin{equation}
{\cal R}_H(-z^*)
=
{\cal R}_H(z)^*,
\qquad \ImPart z\geq0.
\label{eq:RH_reality}
\end{equation}
Boundary values are understood almost everywhere.
Dispersion relations for the admittance require separate subtraction conditions.
\end{defn}

We use the standard Hardy spaces on the upper half-plane
$\mathbb{C}_+=\{z:\ImPart z>0\}$.
The space $H^2(\mathbb{C}_+)$ consists of analytic functions $F$ satisfying
\[
\sup_{y>0}\int_{-\infty}^{\infty}|F(x+iy)|^2\,dx<\infty,
\]
while $H^\infty(\mathbb{C}_+)$ consists of bounded analytic functions on
$\mathbb{C}_+$, with
\[
\|F\|_{H^\infty}
=
\sup_{\ImPart z>0}|F(z)|.
\]
Thus the analytic boundedness condition in Definition~\ref{def:admissible_response}
states that ${\cal R}_H\in H^\infty(\mathbb{C}_+)$.

The $H^\infty$ assumption has a time-domain meaning.
For a classical chiral component of the canonical field, the energy crossing the plane is proportional to $\int |\partial_t\Psi_\pm(t)|^2dt$.
Thus the energy-normalized signals $v_\pm=\partial_t\Psi_\pm$ have the same reflection multiplier as $\Psi_\pm$.
A causal, stationary boundary initially at rest and passive for every incident packet defines a convolution map $C_H$ satisfying
\begin{equation}
v_+=C_Hv_-,\qquad
\|C_Hv_-\|_{L^2}\leq\|v_-\|_{L^2}.
\label{eq:time_domain_passivity}
\end{equation}
On signals supported in $t\geq0$, the Fourier--Laplace transform identifies this map with a bounded multiplier of $H^2(\mathbb C_+)$.
Such multipliers are precisely $H^\infty$ functions.
In particular, if $k_z$ is the Hardy reproducing kernel, then $M_{{\cal R}_H}^*k_z=\overline{{\cal R}_H(z)}k_z$ gives $|{\cal R}_H(z)|\leq\|C_H\|\leq1$.
This connects Definition~\ref{def:admissible_response} to causal packet passivity, not just to a real-axis test.
The coherent response is taken with zero boundary initial excitation; thermal inputs are added separately in Sec.~\ref{sec:quantum_completion}.

The response can also be expressed as the admittance
\begin{equation}
\Lambda(z)=iz\frac{{\cal R}_H(z)-1}{{\cal R}_H(z)+1},
\label{eq:Lambda_from_RH}
\end{equation}
where the denominator is nonzero.
The equivalent retarded self-energy and its dispersion conditions are discussed in Appendix~\ref{app:response}.
In the following we impose bounded causal passivity on ${\cal R}_H$ itself.

\section{Causality and ultraviolet completion}
\label{sec:causal_completion}

\subsection{Obstruction to a global Boltzmann modulus}

Let $\mathbb C_+=\{z:\ImPart z>0\}$ and let $\beta>0$.
The following constraint applies to the amplitude of a stationary response at a fixed matching plane, with the same Killing time at input and output.

\begin{prop}[Causal obstruction to a global Boltzmann modulus]
\label{prop:boltzmann_obstruction}
There is no nonzero ${\cal R}_H\in H^\infty(\mathbb C_+)$ whose boundary values satisfy
\begin{equation}
|{\cal R}_H(\omega)|^2=e^{-\beta\omega}
\label{eq:global_boltzmann_modulus}
\end{equation}
for almost every $\omega>0$.
\end{prop}

\begin{proof}
For a nonzero bounded analytic function, the logarithmic boundary integral satisfies \cite[Chap.~II, Sec.~4]{Garnett2007}
\begin{equation}
\int_{-\infty}^{\infty}
\frac{\mu\log|{\cal R}_H(\omega)|}{\mu^2+\omega^2}\,d\omega>-\infty,
\qquad \mu>0.
\label{eq:Hardy_log_condition}
\end{equation}
On the positive half-axis, Eq.~(\ref{eq:global_boltzmann_modulus}) gives $\log|{\cal R}_H|=-\beta\omega/2$, whose weighted integral diverges to $-\infty$.
The positive part of the integrand is integrable because ${\cal R}_H$ is bounded, so the negative half-axis cannot compensate for this divergence.
\end{proof}

Reality is not needed for this obstruction.
For a real boundary response, it also gives $\log|{\cal R}_H(\omega)|=-\beta|\omega|/2$ on the full real axis.
The conclusion constrains a causal transfer amplitude; a thermal state or a detailed-balance ratio of transition rates is not subject to the same requirement.
In particular, the low-frequency and near-horizon assumptions of Ref.~\cite{OshitaWangAfshordi2020}, Secs.~2.1 and 2.3, do not construct a fixed-plane amplitude at arbitrarily high frequency.
In Ref.~\cite{AmaroSeoane2025}, Sec.~III, the discussion following its Boltzmann reflectivity formula associates $R\leq1$ with causality and stability.
This inequality tests real-frequency gain, whereas Eq.~(\ref{eq:Hardy_log_condition}) is an additional restriction on a global causal response.
The proposition addresses a global extension of that prescription and leaves its finite-band use open.

Kuchiev and Flambaum derive $|R|=\exp[-\omega/(2T_H)]$ by analytically continuing the Schwarzschild scattering problem with $r_g\to r_g+i0$ and matching solutions through the horizon \cite{KuchievFlambaum2005}.
Their causality argument concerns analyticity of the scattering matrix and the continuation prescription.
It does not establish that this horizon coefficient is a bounded Fourier--Laplace multiplier for a causal passive convolution between local ports at a fixed finite plane, measured in the same Killing time.
That stronger property is the hypothesis of Proposition~\ref{prop:boltzmann_obstruction}.
If their modulus is identified with such a response over the entire positive axis, the proposition rules out that identification.
This conditional statement does not diagnose which premise fails in their construction, or invalidate its scattering interpretation.
Relating the two descriptions requires specifying the local ports and verifying the packet-passivity bound in Eq.~(\ref{eq:time_domain_passivity}); analyticity alone does not establish boundedness.

\subsection{Causal spectral factors and exact finite-band matching}

For an even function $0<W(\omega)<1$ almost everywhere, suppose that both $\log W$ and $\log(1-W)$ are integrable with weight $(\mu^2+\omega^2)^{-1}$.
Define the corresponding upper-half-plane outer function \cite[Chap.~II, Sec.~4]{Garnett2007}
\begin{equation}
{\cal O}[u](z)=\exp\left\{
\frac{1}{\pi i}\int_{-\infty}^{\infty}
\left(\frac{1}{\nu-z}-\frac{\nu}{\nu^2+\mu^2}\right)u(\nu)\,d\nu
\right\},\qquad z\in\mathbb C_+ .
\label{eq:outer_function}
\end{equation}
Taking $u_R=\frac12\log W$ and $u_L=\frac12\log(1-W)$ gives bounded analytic amplitudes ${\cal O}[u_R]$ and ${\cal O}[u_L]$ with squared moduli $W$ and $1-W$.
Even data allow reality-preserving phases.
This is a causal output row with unit real-frequency norm.
A complete causal unitary scattering realization additionally includes an environment output; an explicit rational example is given in Appendix~\ref{app:rational}.

For any band edge $\omega_b>0$, consider
\begin{equation}
W_b(\omega)=
\begin{cases}
e^{-\beta|\omega|},& |\omega|\leq\omega_b,\\
e^{-\beta\omega_b}(|\omega|/\omega_b)^{-\beta\omega_b},& |\omega|>\omega_b.
\end{cases}
\label{eq:exact_band_modulus}
\end{equation}
The function and its first derivative match at $|\omega|=\omega_b$.
Its ultraviolet logarithm grows only as $-\beta\omega_b\log(|\omega|/\omega_b)$.
At zero, $\log(1-W_b)=\log(\beta|\omega|)+O(|\omega|)$ is locally integrable.
Both spectral factors therefore exist.
This construction preserves the exact Boltzmann modulus throughout the chosen band.
In the passive KMS channel derived below, it also preserves the thermal equality in that band.
There is no universal nonzero in-band correction forced by Proposition~\ref{prop:boltzmann_obstruction}; the ultraviolet completion still fixes part of the phase.

\subsection{A smooth crossover family}

To examine the dependence on a single crossover scale, we use
\begin{equation}
W_{\rm UV}(\omega)=\left(1+\frac{|\omega|}{\omega_c}\right)^{-\beta\omega_c},
\qquad \omega_c>0.
\label{eq:UV_flux_reflectivity}
\end{equation}
For fixed positive $\omega/\omega_c\ll1$,
\begin{equation}
\log W_{\rm UV}(\omega)
=-\beta\omega+O\left(\frac{\beta\omega^2}{\omega_c}\right).
\label{eq:UV_low_frequency_match}
\end{equation}
The ultraviolet tail is polynomial, so the data
\begin{equation}
u_R(\nu)=\frac12\log W_{\rm UV}(\nu),\qquad
u_L(\nu)=\frac12\log[1-W_{\rm UV}(\nu)]
\label{eq:outer_boundary_data}
\end{equation}
satisfy the required integrability conditions.
We set
\begin{equation}
{\cal R}_{H,{\rm UV}}={\cal O}[u_R],\qquad
{\cal L}_{H,{\rm UV}}={\cal O}[u_L],
\label{eq:UV_outer_pair}
\end{equation}
which gives
\begin{equation}
|{\cal R}_{H,{\rm UV}}|^2=W_{\rm UV},\qquad
|{\cal L}_{H,{\rm UV}}|^2=1-W_{\rm UV}.
\label{eq:UV_passive_pair}
\end{equation}
We select the outer reflection with no additional inner factor and with its continuous phase tending to zero at $\omega=0^+$.
Multiplication by a bounded analytic inner factor, whose boundary modulus is one almost everywhere, such as $e^{i\tau\omega}$ with $\tau\geq0$, preserves the modulus and changes the delay.
Both completions use this outer-response assumption, which is
additional to causality and the prescribed modulus.
We distinguish varying the matching surface from varying the cutoff at a fixed plane.
The cusp at zero, and the band-edge joins in Eq.~(\ref{eq:exact_band_modulus}), preclude a global meromorphic continuation.
The complex-pole statements in Appendix~\ref{app:poles} are restricted to local continuation domains avoiding such points.

\section{Redshift, outer phase, and matching-scale dependence}
\label{sec:matching}

\subsection{Redshift of a local ultraviolet scale}
\label{sec:redshift_local_cutoff}

The crossover scale $\omega_c$ in Eq.~(\ref{eq:UV_flux_reflectivity}) has so far been defined with respect to the Schwarzschild Killing time.
It acquires a direct near-horizon meaning if the effective response is specified in the local proper frame at the matching surface.
Let $r_0>2M$ be the areal radius corresponding to $x_0=r_*(r_0)$ and put
\begin{equation}
f_0=f(r_0).
\label{eq:f0_definition}
\end{equation}
A static local observer has proper time $d\tau=\sqrt{f_0}\,dt$, so a mode with Killing frequency $\omega$ has local frequency
\begin{equation}
\omega_{\rm loc}=\frac{\omega}{\sqrt{f_0}}.
\label{eq:local_frequency_redshift}
\end{equation}
For a KMS state at the Hawking temperature
$\beta_H=8\pi M$ \cite{Hawking1975}
with respect to Killing time, the Tolman--Ehrenfest
relation \cite{TolmanEhrenfest1930} gives
\begin{equation}
\beta_{\rm loc}=\beta_H\sqrt{f_0},
\qquad
\beta_H=8\pi M.
\label{eq:Tolman_beta}
\end{equation}

Suppose that the smooth local flux reflectivity is written in the proper frame as
\begin{equation}
W_{\rm loc}(\omega_{\rm loc})
=\left(1+\frac{|\omega_{\rm loc}|}{\Omega_{\rm UV}}\right)^{-\beta_{\rm loc}\Omega_{\rm UV}},
\label{eq:local_UV_reflectivity}
\end{equation}
where $\Omega_{\rm UV}$ is a fixed local crossover scale.
Then its expression in Killing frequency is exactly Eq.~(\ref{eq:UV_flux_reflectivity}) with $\beta=\beta_H$ and
\begin{equation}
\omega_c(r_0)=\sqrt{f_0}\,\Omega_{\rm UV}.
\label{eq:redshifted_cutoff}
\end{equation}

This follows directly from $\beta_{\rm loc}\omega_{\rm loc}=\beta_H\omega$.
For the exact-band prescription, a fixed local band edge $\Omega_b$ similarly gives
\begin{equation}
\omega_b(r_0)=\sqrt{f_0}\,\Omega_b.
\label{eq:redshifted_band_edge}
\end{equation}
The Tolman relation fixes the redshift of an equilibrium temperature.
Its value at infinity is a separate input, which we set to the Hawking temperature.
Neither assumption derives the reflectivity law.

To display the near-horizon scaling, define the proper distance from the event horizon to the matching surface by
\begin{equation}
d_0=\int_{2M}^{r_0}\frac{dr}{\sqrt{f(r)}}
\label{eq:proper_matching_distance}
\end{equation}
and set
\begin{equation}
y=\sqrt{\frac{r_0-2M}{2M}},
\qquad
\zeta=\frac{d_0}{4M}.
\label{eq:y_zeta_definition}
\end{equation}
The Schwarzschild geometry gives the exact relations
\begin{eqnarray}
d_0&=&2M\left[y\sqrt{1+y^2}+\operatorname{arcsinh}y\right],
\label{eq:proper_distance_exact}\\
\sqrt{f_0}&=&\frac{y}{\sqrt{1+y^2}},
\label{eq:redshift_y}\\
x_0&=&2M\left(1+y^2+2\log y\right).
\label{eq:x0_y_exact}
\end{eqnarray}
Series inversion at $\zeta\ll1$ gives
\begin{eqnarray}
\sqrt{f_0}
&=&\zeta-\frac{2}{3}\zeta^3+\frac{11}{15}\zeta^5+O(\zeta^7),
\label{eq:redshift_proper_series}\\
\frac{x_0}{M}
&=&2+4\log\zeta+\frac{4}{3}\zeta^2-\frac{13}{45}\zeta^4+O(\zeta^6).
\label{eq:tortoise_proper_series}
\end{eqnarray}
Consequently,
\begin{equation}
\omega_c
=\frac{d_0\Omega_{\rm UV}}{4M}
\left[1-\frac{2}{3}\left(\frac{d_0}{4M}\right)^2
+O\left(\frac{d_0^4}{M^4}\right)\right].
\label{eq:cutoff_proper_distance}
\end{equation}
Thus the dimensionless combination $d_0\Omega_{\rm UV}$ controls the leading redshifted crossover.
This relation also shows why $\omega_c$ cannot be held fixed along this response family when the matching plane is varied at fixed local cutoff.

\subsection{Outer phases of the two completions}
\label{sec:UV_phase}

For even boundary data the continuous outer phase is
\begin{equation}
\phi_c(\omega)=\frac{2\omega}{\pi}
\int_0^\infty\frac{u_R(\nu)-u_R(\omega)}{\omega^2-\nu^2}\,d\nu,
\qquad \omega>0.
\label{eq:outer_phase_quadrature}
\end{equation}
The singularity at $\nu=\omega$ is removable.
The subscript $c$ denotes the scale $\omega_c$.

\begin{prop}[Dispersive delay for the smooth completion]
\label{prop:UV_phase_delay}
For Eq.~(\ref{eq:UV_flux_reflectivity}) with the specified outer phase,
\begin{eqnarray}
\tau_H(\omega;\omega_c):=\frac{d\phi_c}{d\omega}
&=&\frac{\beta}{\pi}
\frac{\log(\omega_c/\omega)}{1-(\omega/\omega_c)^2},
\label{eq:UV_group_delay}\\
\phi_c(\omega)&=&\frac{\beta\omega_c}{\pi}
\int_0^{\omega/\omega_c}\frac{\log(1/t)}{1-t^2}\,dt.
\label{eq:UV_phase_integral}
\end{eqnarray}
At $\omega=\omega_c$, the removable limit of $\tau_H$ is $\beta/(2\pi)$.
In particular, $\tau_H>0$ for $\omega>0$.
\end{prop}

\begin{proof}
Differentiating the Hilbert-transform boundary value and integrating by parts gives
\begin{equation}
\phi_c'(\omega)=\frac{\beta\omega_c}{\pi}\PV
\int_0^\infty\frac{\nu\,d\nu}{(\nu+\omega_c)(\nu^2-\omega^2)}.
\label{eq:UV_phase_derivative_integral}
\end{equation}
This follows first on finite intervals with a symmetric exclusion at $\nu=\omega$; the differentiated data decay as $1/|\nu|$, so the resulting tail is integrable.
For $\omega\neq\omega_c$, partial fractions give
\begin{equation}
\frac{\nu}{(\nu+\omega_c)(\nu^2-\omega^2)}
=-\frac{\omega_c}{\omega_c^2-\omega^2}\frac{1}{\nu+\omega_c}
+\frac{1}{2(\omega_c+\omega)}\frac{1}{\nu-\omega}
+\frac{1}{2(\omega_c-\omega)}\frac{1}{\nu+\omega}.
\label{eq:UV_phase_partial_fractions}
\end{equation}
The logarithms at infinity cancel, leaving
$\omega_c\log(\omega_c/\omega)/(\omega_c^2-\omega^2)$ for the integral.
This proves Eq.~(\ref{eq:UV_group_delay}).
Integration from zero, using $\phi_c(0^+)=0$, gives Eq.~(\ref{eq:UV_phase_integral}).
\end{proof}

For $0<\omega/\omega_c<1$, expansion of the last integrand yields the convergent series
\begin{equation}
\phi_c(\omega)=\frac{\beta\omega}{\pi}
\sum_{j=0}^{\infty}\left(\frac{\omega}{\omega_c}\right)^{2j}
\left[\frac{\log(\omega_c/\omega)}{2j+1}+\frac{1}{(2j+1)^2}\right].
\label{eq:UV_phase_series}
\end{equation}

For the exact-band modulus in Eq.~(\ref{eq:exact_band_modulus}), let $\phi_b$ denote the outer phase normalized by $\phi_b(0^+)=0$.

\begin{prop}[Outer phase for exact finite-band matching]
\label{prop:band_phase}
For $0<\omega<\omega_b$, put $s=\omega/\omega_b$.
The outer phase satisfies
\begin{equation}
\frac{d\phi_b}{d\omega}
=\frac{\beta}{\pi}\left[
\log\frac{\omega_b}{\omega}
+\frac12\log(1-s^2)
+\frac{1}{2s}\log\frac{1+s}{1-s}\right].
\label{eq:band_group_delay}
\end{equation}
Its convergent expansion is
\begin{equation}
\phi_b(\omega)=\frac{\beta\omega}{\pi}
\left[\log\frac{\omega_b}{\omega}+2
-\sum_{j=1}^{\infty}\frac{s^{2j}}{2j(2j+1)^2}\right].
\label{eq:band_phase_series}
\end{equation}
In particular, as $\omega/\omega_b\to0$,
\begin{equation}
\phi_b(\omega)
=\frac{\beta\omega}{\pi}\left[\log\frac{\omega_b}{\omega}+2\right]
-\frac{\beta\omega^3}{18\pi\omega_b^2}
+O\left(\frac{\beta\omega^5}{\omega_b^4}\right).
\label{eq:band_phase_expansion}
\end{equation}
\end{prop}

\begin{proof}
For positive $\nu$, the logarithmic boundary datum has derivative $u_R'(\nu)=-\beta/2$ below $\omega_b$ and $u_R'(\nu)=-\beta\omega_b/(2\nu)$ above it.
Differentiating the subtracted Hilbert transform as in Eq.~(\ref{eq:UV_phase_derivative_integral}) gives
\begin{equation}
\phi_b'(\omega)=\frac{\beta}{\pi}\left[
\PV\int_0^{\omega_b}\frac{\nu\,d\nu}{\nu^2-\omega^2}
+\omega_b\int_{\omega_b}^{\infty}\frac{d\nu}{\nu^2-\omega^2}\right].
\label{eq:band_phase_derivative_integral}
\end{equation}
The two integrals equal $\log(\omega_b/\omega)+\frac12\log(1-s^2)$ and $(2s)^{-1}\log[(1+s)/(1-s)]$, respectively.
This proves Eq.~(\ref{eq:band_group_delay}).
Expanding the logarithms gives
$\phi_b'=(\beta/\pi)[\log(\omega_b/\omega)+1-\sum_{j\geq1}s^{2j}/(2j(2j+1))]$.
Integration from zero, with termwise integration on compact subintervals of $0\leq s<1$, proves Eqs.~(\ref{eq:band_phase_series}) and (\ref{eq:band_phase_expansion}).
\end{proof}

\subsection{A class of completions with a common low-frequency slope}
\label{sec:general_outer_phase}

The common logarithmic coefficient does not require either of the two specified tails.
To state sufficient conditions, write the flux reflectivity as
\begin{equation}
W_h(\omega)=\exp\left[-\beta\omega_c h\left(\frac{|\omega|}{\omega_c}\right)\right],
\qquad \beta>0,\quad \omega_c>0,
\label{eq:profile_modulus}
\end{equation}
where the dimensionless profile $h$ is held fixed when the cutoff is varied.
The normalization $h'(0)=1$ fixes the low-frequency Boltzmann slope.

\begin{prop}[Outer phase for a class of ultraviolet completions]
\label{prop:general_outer_phase}
Let $h:[0,\infty)\to[0,\infty)$ be measurable, with $h\in C^3([0,1])$, $h(0)=0$, $h'(0)=1$, and
\begin{equation}
\int_1^\infty\frac{h(t)}{t^2}\,dt<\infty.
\label{eq:profile_tail_condition}
\end{equation}
Then Eq.~(\ref{eq:profile_modulus}) admits a bounded outer reflection amplitude.
For its phase $\phi_h$, normalized by $\phi_h(0^+)=0$, put $s=\omega/\omega_c$.
As $s\to0^+$,
\begin{equation}
\phi_h(\omega)=\frac{\beta\omega}{\pi}
\left[\log\frac{\omega_c}{\omega}+C_h
+O_h\left(s^2[1+\log(1/s)]\right)\right],
\label{eq:profile_phase_expansion}
\end{equation}
where
\begin{equation}
C_h=\int_0^1\frac{h(t)-t}{t^2}\,dt
+\int_1^\infty\frac{h(t)}{t^2}\,dt.
\label{eq:profile_phase_constant}
\end{equation}
The remainder constant depends on the fixed profile $h$.
If $h(t)=t$ on $[0,1]$, the remainder in brackets improves to $O_h(s^2)$.
\end{prop}

\begin{proof}
The stated assumptions make $\log W_h$ integrable with the weight used in Eq.~(\ref{eq:outer_function}).
Since $\log W_h\leq0$, its outer spectral factor is bounded by one.
For $0<s<1$, rescaling the Hilbert transform in Eq.~(\ref{eq:outer_phase_quadrature}) gives
\begin{equation}
\phi_h(\omega)=\frac{\beta\omega}{\pi}J_h(s),
\qquad
J_h(s)=\PV\int_0^\infty\frac{h(t)}{t^2-s^2}\,dt.
\label{eq:profile_phase_integral}
\end{equation}
For $0\leq t\leq1$, write $h(t)=t+a t^2+t^2q(t)$, where $a=h''(0)/2$.
The regularity assumption implies $q(0)=0$ and $q\in C^1([0,1])$.
The first two terms give
\begin{eqnarray}
\PV\int_0^1\frac{t\,dt}{t^2-s^2}
&=&\log(1/s)+\frac12\log(1-s^2),\nonumber\\
\PV\int_0^1\frac{t^2\,dt}{t^2-s^2}
&=&1+\frac{s}{2}\log\frac{1-s}{1+s}.
\label{eq:profile_elementary_integrals}
\end{eqnarray}
After subtracting $\int_0^1q(t)\,dt$, the remaining contribution from $t^2q(t)$ is $s^2\PV\int_0^1q(t)/(t^2-s^2)\,dt$.
Subtracting $q(s)$ in this last integral bounds its nonsingular part by
$\|q'\|_\infty\int_0^1dt/(t+s)=O_h(1+\log(1/s))$.
The constant part is $q(s)(2s)^{-1}\log[(1-s)/(1+s)]$ and is bounded as $s\to0^+$.
Finally, Eq.~(\ref{eq:profile_tail_condition}) gives
\begin{equation}
\int_1^\infty\frac{h(t)}{t^2-s^2}\,dt
=\int_1^\infty\frac{h(t)}{t^2}\,dt+O_h(s^2).
\label{eq:profile_tail_expansion}
\end{equation}
Combining these terms proves Eqs.~(\ref{eq:profile_phase_expansion}) and (\ref{eq:profile_phase_constant}).
If $h(t)=t$ on $[0,1]$, then $a=q=0$, which gives the sharper remainder.
\end{proof}

The smooth completion has $h(t)=\log(1+t)$ and $C_h=1$.
The exact-band completion has $h(t)=t$ for $t\leq1$ and $h(t)=1+\log t$ for $t>1$, with $\omega_c$ replaced by $\omega_b$ and $C_h=2$.
Thus Eqs.~(\ref{eq:UV_phase_series}) and (\ref{eq:band_phase_series}) refine the general result for these profiles.
Proposition~\ref{prop:general_outer_phase} concerns outer reflection amplitudes with a fixed profile; it does not constrain additional inner factors or arbitrary cutoff-dependent changes of shape.

\subsection{Reference-plane transport and a family of local responses}
\label{sec:reference_transport}

Two operations must be distinguished.
For one fixed scatterer, changing the reference plane from $x_1$ to $x_2$ within a free region gives
\begin{equation}
{\cal R}_H(\omega;x_2)
=e^{2i\omega(x_2-x_1)}{\cal R}_H(\omega;x_1).
\label{eq:reference_plane_transport}
\end{equation}
Thus ${\cal R}_H(\omega;x)e^{-2i\omega x}$ is exactly unchanged.
If a nonzero potential is crossed, the same boundary data must instead be transported by the wave equation; where defined, their admittance obeys
\begin{equation}
\partial_x\Lambda(\omega;x)
=V(x)-\omega^2-\Lambda(\omega;x)^2.
\label{eq:admittance_transport}
\end{equation}
This is the matching operation underlying boundary effective theories \cite{Burgess2018}.

Our second operation assigns a specified local response law at each $r_0$, holds its proper-frame cutoff fixed, and applies Eqs.~(\ref{eq:Tolman_beta}), (\ref{eq:redshifted_cutoff}), or (\ref{eq:redshifted_band_edge}).
For the smooth family even the modulus then changes with $r_0$.
The family is not asserted to solve Eq.~(\ref{eq:admittance_transport}); it need not represent exact descriptions of one physical scatterer.
The following result identifies the leading matching-distance dependence that cancels within this family and the residual dependence that remains.

\subsection{Leading cancellation and completion-dependent remainders}

For Eq.~(\ref{eq:profile_modulus}), the local law $W_{h,{\rm loc}}=\exp[-\beta_{\rm loc}\Omega_{\rm UV}h(|\omega_{\rm loc}|/\Omega_{\rm UV})]$ gives the same redshifted cutoff as Eq.~(\ref{eq:redshifted_cutoff}).
The phase expansion in Proposition~\ref{prop:general_outer_phase} therefore applies to the matching-scale comparison with a fixed local profile.

\begin{prop}[Near-horizon matching-scale cancellation]
\label{prop:matching_scale_cancellation}
For a fixed profile $h$ satisfying Proposition~\ref{prop:general_outer_phase}, take its outer reflection with no additional inner factor.
Let $\beta=\beta_H=8\pi M$, let $\Omega_{\rm UV}$ be fixed in the local proper frame, and let $\omega_c(\zeta)$ be given by Eq.~(\ref{eq:redshifted_cutoff}).
Define the part of the round-trip phase assigned to near-horizon propagation and reflection by
\begin{equation}
\Phi_{\rm in}(\omega;\zeta)
:=\phi_h(\omega;\omega_c(\zeta))-2\omega x_0(\zeta).
\label{eq:inner_roundtrip_phase}
\end{equation}
In the overlap regime
\begin{equation}
\zeta\ll1,
\qquad
0<\frac{\omega}{\omega_c(\zeta)}\ll1,
\label{eq:matching_overlap_regime}
\end{equation}
the leading logarithmic dependence on the matching scale cancels and
\begin{equation}
\begin{split}
\Phi_{\rm in}(\omega;\zeta)
={}&\Phi_0(\omega)
-8M\omega\zeta^2+O(M\omega\zeta^4)\\
&+O_h\left[M\omega\left(\frac{\omega}{\omega_c}\right)^2
\left(1+\log\frac{\omega_c}{\omega}\right)\right],
\end{split}
\label{eq:matching_scale_cancellation}
\end{equation}
where
\begin{equation}
\Phi_0(\omega)=8M\omega\log\frac{\Omega_{\rm UV}}{\omega}
+(8C_h-4)M\omega
\label{eq:profile_leading_roundtrip}
\end{equation}
is independent of $\zeta$ in the tortoise convention of Eq.~(\ref{eq:x0_y_exact}).
If $h(t)=t$ on $[0,1]$, the dispersive remainder improves to $O_h[M\omega(\omega/\omega_c)^2]$.
Thus the term proportional to $M\omega\log\zeta$ cancels at leading order in $\omega/\omega_c$; logarithmic dependence may remain in the suppressed dispersive correction.
\end{prop}

\begin{proof}
Equations~(\ref{eq:redshift_proper_series}) and (\ref{eq:redshifted_cutoff}) give
\begin{equation}
\log\omega_c
=\log\Omega_{\rm UV}+\log\zeta-\frac{2}{3}\zeta^2+O(\zeta^4).
\label{eq:log_cutoff_proper_series}
\end{equation}
Proposition~\ref{prop:general_outer_phase} therefore gives
\begin{equation}
\begin{split}
\phi_h(\omega)
={}&8M\omega\left[\log\frac{\Omega_{\rm UV}}{\omega}
+\log\zeta+C_h-\frac{2}{3}\zeta^2+O(\zeta^4)\right]\\
&+O_h\left[M\omega\left(\frac{\omega}{\omega_c}\right)^2
\left(1+\log\frac{\omega_c}{\omega}\right)\right].
\end{split}
\label{eq:phase_matching_expansion}
\end{equation}
On the other hand, Eq.~(\ref{eq:tortoise_proper_series}) gives
\begin{equation}
-2\omega x_0
=-8M\omega\log\zeta-\frac{8}{3}M\omega\zeta^2
-4M\omega+O(M\omega\zeta^4).
\label{eq:propagation_matching_expansion}
\end{equation}
The $\log\zeta$ terms cancel exactly and the two quadratic terms sum to $-8M\omega\zeta^2$.
The remaining linear term gives Eq.~(\ref{eq:profile_leading_roundtrip}), and the dispersive bound follows from Proposition~\ref{prop:general_outer_phase}.
\end{proof}

For two matching surfaces $\zeta_1$ and $\zeta_2$ satisfying Eq.~(\ref{eq:matching_overlap_regime}), Proposition~\ref{prop:matching_scale_cancellation} gives
\begin{equation}
\begin{split}
\Phi_{\rm in}(\omega;\zeta_2)-\Phi_{\rm in}(\omega;\zeta_1)
={}&-8M\omega(\zeta_2^2-\zeta_1^2)
+O[M\omega(\zeta_1^4+\zeta_2^4)]\\
&+O_h\left[M\omega\sum_{i=1}^{2}
\left(\frac{\omega}{\omega_c(\zeta_i)}\right)^2
\left(1+\log\frac{\omega_c(\zeta_i)}{\omega}\right)\right].
\end{split}
\label{eq:matching_scale_difference}
\end{equation}
The cancellation concerns the leading phase within the specified family.
It does not prove that the complete response or the emitted spectrum is independent of the matching surface.
In particular, the dispersive term in Eq.~(\ref{eq:matching_scale_cancellation}) and the change in the smooth modulus remain.
The geometric and dispersive expansions require two simultaneous smallness conditions.
At fixed $\omega$ and $\Omega_{\rm UV}$, these give the overlap interval $\omega/\Omega_{\rm UV}\ll\zeta\ll1$ to leading order.
The two remainders are kept separately, and neither is assumed smaller than the other.

The overlap condition has a direct physical meaning:
\begin{equation}
\frac{\omega}{\omega_c}
=\frac{4M\omega}{d_0\Omega_{\rm UV}}\left[1+O(\zeta^2)\right]\ll1.
\label{eq:physical_cutoff_hierarchy}
\end{equation}
At the representative ringdown-band value $M\omega=0.35$, it requires $d_0\Omega_{\rm UV}\gg1.4$.
If the proper distance is of order the Planck length $\ell_{\rm Pl}$ and the local cutoff is of order $\ell_{\rm Pl}^{-1}$, their product is of order unity.
The ringdown band is then not parametrically within this low-energy regime.
The expansion still applies at lower frequencies satisfying $4M\omega\ll d_0\Omega_{\rm UV}$.
The same condition holds for the exact-band scale after replacing $\Omega_{\rm UV}$ by $\Omega_b$.

For the smooth completion, $C_h=1$ gives
\begin{equation}
\Phi_{0,c}(\omega)
=8M\omega\log\frac{\Omega_{\rm UV}}{\omega}+4M\omega.
\label{eq:leading_smooth_roundtrip}
\end{equation}
For this completion, $\Phi_0=\Phi_{0,c}$ and $\Phi_{{\rm in},c}=\Phi_{\rm in}$.
The corresponding calculation for Eq.~(\ref{eq:band_phase_expansion}), with $\beta=8\pi M$ and $\Omega_b$ fixed, gives
\begin{eqnarray}
\Phi_{{\rm in},b}(\omega;\zeta)
&:=&\phi_b(\omega;\omega_b(\zeta))-2\omega x_0(\zeta),\nonumber\\
&=&\Phi_{0,b}(\omega)-8M\omega\zeta^2
-\frac{4M\omega^3}{9\omega_b(\zeta)^2}\nonumber\\
&&+O(M\omega\zeta^4)
+O\left[ M\omega\left(\frac{\omega}{\omega_b}\right)^4\right],
\label{eq:band_matching_cancellation}\\
\Phi_{0,b}(\omega)
&=&8M\omega\log\frac{\Omega_b}{\omega}+12M\omega.
\label{eq:leading_band_roundtrip}
\end{eqnarray}
This holds for $\zeta\ll1$ and $0<\omega/\omega_b(\zeta)\ll1$.
Thus exact Boltzmann matching in a finite band is compatible with the same leading cancellation.
The different linear constant and dispersive terms record the choice of completion.
Adding an inner factor $I(z;\zeta)$ would add $\arg I(\omega;\zeta)$ to these phases.
A matching-dependent inner delay can change their scale dependence and is not constrained by the modulus alone.

For a general fixed inverse Killing temperature $\beta$, the leading logarithmic dependence on the matching distance for any fixed profile in Proposition~\ref{prop:general_outer_phase} is
\begin{equation}
\left(\frac{\beta}{\pi}-8M\right)\omega\log\zeta.
\label{eq:temperature_mismatch_log}
\end{equation}
Thus the cancellation in these families requires the Hawking value $\beta=8\pi M$, in addition to the Tolman redshift relation.

In the near-horizon composition of Sec.~\ref{sec:composition}, the full round-trip phase is $\arg{\cal R}_{\bh}+\Phi_{\rm in}$.
Using the common horizon-asymptotic barrier leaves the above leading cancellation unchanged.
Exact scattering at a finite plane also changes the assigned exterior-potential tail; this effect is distinct from the two response remainders.
Appendix~\ref{app:numerical} checks the finite-plane approximation on the frequency grid used for the fixed-plane spectra.

\begin{rem}[Event-horizon limit and domain of validity]
\label{rem:event_horizon_limit}
Equation~(\ref{eq:cutoff_proper_distance}) gives $\omega_c\to0$ as $d_0\to0$ for fixed finite $\Omega_{\rm UV}$.
Hence the limit $d_0\to0$ at fixed nonzero Killing frequency does not remain in Eq.~(\ref{eq:matching_overlap_regime}); equivalently, $\omega_{\rm loc}=\omega/\sqrt{f_0}$ diverges.
Proposition~\ref{prop:matching_scale_cancellation} is therefore an overlap statement for a stretched-horizon effective response, not a uniform extrapolation of a finite-local-cutoff model to the mathematical event horizon.
The loss of the low-energy expansion in that strict limit is the expected gravitational blueshift, not an inconsistency of the matching construction.
\end{rem}

\subsection{Direct comparison at fixed local cutoff}

Figure~\ref{fig:matching_scale} varies the matching distance while holding the proper-frame cutoff fixed.
Its upper panel separates the propagation and smooth outer-phase changes relative to a reference surface.
Their sum has no leading logarithmic change.
The lower panel resolves the remaining phase after subtracting the scale-independent leading term, for both completions.
The geometric and dispersive corrections can dominate in different parts of the same overlap interval.
This plot evaluates the response-family comparison directly; it does not add or fit a delay.

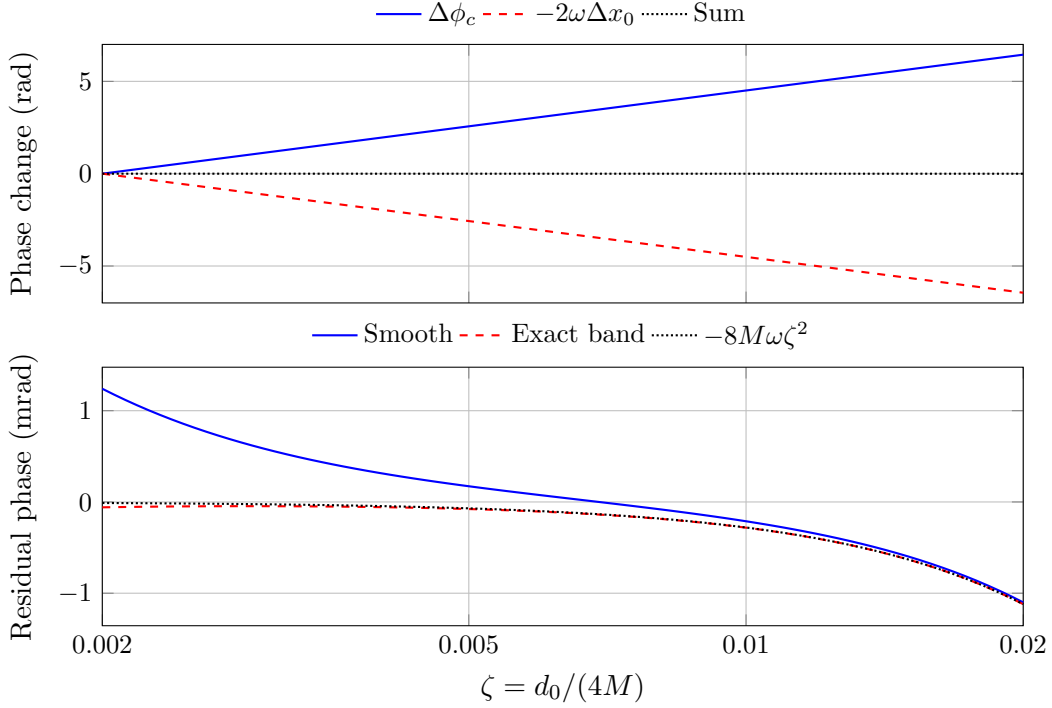
\begin{figure}[htbp]
\centering
\pgfplotstableread[col sep=space]{
zeta propagation reflection total smooth band geometry
0.002 0 0 0 1.240661050681 -0.05884174322499 -0.0112
0.002058401054389 -0.0805906994683 0.08052776449914 -6.293496916517e-05 1.177726081516 -0.05684021842889 -0.01186364172199
0.002118507450355 -0.1611814120443 0.1610585527885 -0.0001228592558535 1.117801794828 -0.05502711847531 -0.01256660668818
0.00218036898477 -0.2417721385047 0.2415922168956 -0.0001799216090959 1.060739441586 -0.05339643109013 -0.01331122494729
0.002244036908604 -0.3223628796721 0.3221286155981 -0.0002342640740451 1.006396976636 -0.05194274889555 -0.01409996461209
0.002309563969379 -0.4029536364181 0.4026676141102 -0.000286022307932 0.9546387427495 -0.05066125145627 -0.01493544004023
0.002377004454874 -0.4835444096657 0.4832090837823 -0.0003353258834356 0.9053351672459 -0.04954768927434 -0.01582042049977
0.0024464142381 -0.5641352003926 0.5637529018127 -0.0003822985799981 0.8583624706834 -0.04859836967928 -0.01675783934826
0.002517850823588 -0.6447260096347 0.6442989509711 -0.0004270586636042 0.8136023870772 -0.04781014456639 -0.01775080375556
0.002591373395034 -0.7253168384889 0.7248471193334 -0.0004697191555293 0.7709418951522 -0.04718039994225 -0.01880260500297
0.002667042864327 -0.8059076881174 0.8053973000268 -0.0005103880905476 0.7302729601339 -0.04670704724259 -0.01991672939244
0.002744921922015 -0.8864985597511 0.885949390986 -0.0005491687650778 0.6914922856037 -0.04638851639344 -0.02109686980229
0.002825075089246 -0.9670894546938 0.9665032947181 -0.0005861599757329 0.6545010749486 -0.04622375059238 -0.02234693792765
0.002907568771215 -1.047680374327 1.047058918078 -0.000621456248725 0.6192048019565 -0.04621220279239 -0.02367107724617
0.002992471312189 -1.128271320113 1.127616172052 -0.0006551480605684 0.5855129901131 -0.0463538338765 -0.02507367675197
0.003079853052119 -1.208862293602 1.208174971552 -0.0006873220505096 0.5533390001719 -0.04664911251693 -0.02655938550341
0.003169786384922 -1.289453296435 1.28873523521 -0.0007180612251034 0.522599825578 -0.04709901671827 -0.02813312803291
0.003262345818456 -1.370044330352 1.369296885197 -0.0007474451553445 0.4932158953369 -0.04770503704951 -0.02980012066975
0.003357608036245 -1.450635397193 1.449859847027 -0.0007755501667515 0.4651108839299 -0.04846918157559 -0.03156588883016
0.003455651961016 -1.531226498911 1.530424049388 -0.0008024495227947 0.4382115278868 -0.04939398250472 -0.03343628533188
0.003556558820078 -1.61181763757 1.610989423968 -0.000828213602047 0.4124474486345 -0.05048250457332 -0.03541750979389
0.003660412212622 -1.692408815361 1.691555905291 -0.0008529100694305 0.387750981251 -0.05173835519632 -0.03751612918568
0.00376729817898 -1.773000034601 1.772123430559 -0.0008766040419233 0.3640570087582 -0.05316569641627 -0.03973909959416
0.003877305271904 -1.853591297747 1.852691939498 -0.000899358249084 0.3413028015975 -0.05476925869092 -0.04209378928031
0.003990524629938 -1.934182607401 1.933261374213 -0.0009212331887442 0.3194278619373 -0.05655435656468 -0.04458800310199
0.004107050052914 -2.014773966318 2.01383167904 -0.0009422872782146 0.2983737724669 -0.05852690627592 -0.047230008384
0.004226978079673 -2.095365377417 2.094402800416 -0.0009625770013431 0.2780840493384 -0.06069344535836 -0.05002856232091
0.004350408068039 -2.175956843791 2.174974686739 -0.00098215705176 0.2585039989215 -0.06306115430126 -0.05299294100369
0.004477442277137 -2.256548368713 2.25554728824 -0.00100108047264 0.239580578041 -0.06563788034034 -0.05613297016625
0.004608185952112 -2.337139955654 2.33612055686 -0.001019398793311 0.2212622573704 -0.0684321634579 -0.05945905775387
0.004742747411323 -2.417731608287 2.416694446124 -0.001037162163024 0.2034988876571 -0.07145326467839 -0.06298222842132
0.004881238136084 -2.498323330507 2.497268911025 -0.001054419482222 0.1862415684592 -0.07471119675287 -0.06671416007525
0.005023772863019 -2.578915126435 2.577843907904 -0.00107121853161 0.1694425190714 -0.07821675733414 -0.07066722258178
0.005170469679124 -2.65950700044 2.65841939434 -0.001087606099358 0.1530549513235 -0.08198156475212 -0.07485451876768
0.005321450119598 -2.740098957147 2.73899532904 -0.00110362810675 0.1370329439311 -0.08601809650796 -0.07928992785102
0.005476839268529 -2.820691001456 2.819571671724 -0.001119329732603 0.1213313180789 -0.09033973061432 -0.08398815144524
0.005636765862529 -2.901283138559 2.900148383022 -0.001134755536764 0.1059055139174 -0.09496078991846 -0.08896476228912
0.005801362397386 -2.981875373953 2.98072542437 -0.001149949583033 0.09071146764879 -0.09989658955489 -0.09423625586426
0.005970765237836 -3.062467713463 3.061302757902 -0.001164955561805 0.07570548887686 -0.1051634876847 -0.09982010507098
0.006145114730535 -3.143060163258 3.141880346345 -0.001179816912789 0.06084413789248 -0.1107789396889 -0.105734818144
0.006324555320337 -3.223652729871 3.222458152923 -0.001194576948121 0.04608410256025 -0.1167615559962 -0.112
0.006509235669961 -3.304245420226 3.30303614125 -0.001209278976214 0.03138207446776 -0.1231311637346 -0.1186364172199
0.006699308783157 -3.384838241653 3.383614275226 -0.001223966426689 0.01669462399217 -0.1299088724131 -0.1256660668818
0.006894932131463 -3.465431201919 3.464192518942 -0.00123868297675 0.00197807393188 -0.1371171438485 -0.1331122494729
0.007096267784672 -3.54602430925 3.544770836571 -0.001253472679339 -0.01281162865711 -0.1447798665695 -0.1409996461209
0.007303482545097 -3.62661757236 3.625349192267 -0.001268380093471 -0.02771904278977 -0.1529224349452 -0.1493544004023
0.007516748085769 -3.70721100048 3.705927550063 -0.001283450417107 -0.04278936642506 -0.1615718332966 -0.1582042049977
0.007736241092661 -3.787804603385 3.786505873762 -0.001298729622956 -0.0580685722744 -0.1707567252725 -0.1675783934826
0.00796214341107 -3.868398391433 3.867084126835 -0.00131426459763 -0.07360354694872 -0.180507548783 -0.1775080375556
0.008194642196271 -3.948992375593 3.947662272308 -0.001330103284542 -0.0894422338609 -0.1908566168044 -0.1880260500297
0.008433930068572 -4.029586567484 4.028240272653 -0.001346294830997 -0.1056337803152 -0.2018382243901 -0.1991672939244
0.008680205272895 -4.110180979416 4.108818089676 -0.001362889739912 -0.1222286892307 -0.2134887622384 -0.2109686980229
0.008933671843019 -4.190775624424 4.189395684397 -0.001379940026646 -0.139278975965 -0.2258468371949 -0.2234693792765
0.009194539770617 -4.271370516319 4.269973016938 -0.001397499381401 -0.1568383307194 -0.2389534000856 -0.2367107724617
0.00946302517923 -4.35196566973 4.350550046392 -0.001415623337714 -0.1749622870326 -0.2528518813025 -0.2507367675197
0.009739350503317 -4.432561100149 4.431126730701 -0.001434369447567 -0.193708396886 -0.2675883345913 -0.2655938550341
0.01002374467255 -4.51315682399 4.511703026527 -0.001453797463652 -0.2131364129706 -0.2832115895124 -0.2813312803291
0.01031644330145 -4.593752858638 4.592278889109 -0.001473969529372 -0.2333084786908 -0.2997734130806 -0.2980012066975
0.01061768888462 -4.674349222507 4.67285427213 -0.001494950377185 -0.2542893265037 -0.3173286811156 -0.3156588883016
0.01092773099764 -4.754945935104 4.753429127568 -0.001516807535908 -0.2761464852263 -0.3359355598677 -0.3343628533188
0.01124682650381 -4.835543017088 4.83400340554 -0.001539611547651 -0.2989504969692 -0.3556556985179 -0.3541750979389
0.01157523976698 -4.916140490346 4.914577054151 -0.001563436195073 -0.3227751443912 -0.376554433186 -0.3751612918568
0.01191324287058 -4.996738378059 4.99515001932 -0.001588358739683 -0.347697689002 -0.3987010031186 -0.3973909959416
0.012261115843 -5.077336704783 5.075722244611 -0.001614460171961 -0.3737991212798 -0.4221687797693 -0.4209378928031
0.0126191468896 -5.157935496526 5.156293671052 -0.001641825474089 -0.4011644234078 -0.4470355095232 -0.4458800310199
0.01298763263152 -5.238534780839 5.236864236943 -0.001670543896158 -0.4298828454765 -0.4733835708661 -0.47230008384
0.01336687835137 -5.319134586905 5.317433877659 -0.001700709246727 -0.4600481960458 -0.5013002468419 -0.5002856232091
0.01375719824618 -5.399734945635 5.398002525436 -0.001732420198686 -0.4917591480047 -0.5308780136942 -0.5299294100369
0.01415891568768 -5.48033588977 5.478570109158 -0.001765780611401 -0.5251195607197 -0.562214846641 -0.5613297016625
0.01457236349026 -5.560937453992 5.559136554122 -0.001800899870198 -0.5602388195161 -0.5954145437837 -0.5945905775387
0.01499788418665 -5.641539675036 5.639701781792 -0.001837893244275 -0.5972321935935 -0.6305870692157 -0.6298222842132
0.0154358303117 -5.722142591815 5.720265709551 -0.001876882264218 -0.6362212135367 -0.6678489164528 -0.6671416007525
0.01588656469449 -5.802746245544 5.800828250424 -0.001917995120329 -0.6773340696473 -0.7073234933785 -0.7066722258178
0.01635046075887 -5.883350679882 5.881389312799 -0.001961367083073 -0.7207060323912 -0.7491415299625 -0.7485451876768
0.0168279028329 -5.96395594107 5.961948800123 -0.002007140947007 -0.7664798963257 -0.793441510087 -0.7928992785102
0.0173192864672 -6.044562078092 6.042506610593 -0.002055467499631 -0.8148064489497 -0.8403701288928 -0.8398815144524
0.01782501876267 -6.12516914283 6.123062636814 -0.002106506016681 -0.8658449659994 -0.8900827771385 -0.8896476228912
0.0183455187078 -6.20577719024 6.203616765455 -0.002160424785478 -0.9197637347963 -0.9427440541527 -0.9423625586426
0.01888121752572 -6.286386278531 6.284168876873 -0.00221740165803 -0.9767406073485 -0.9985283110538 -0.9982010507098
0.01943255903154 -6.366996469362 6.364718844726 -0.002277624635679 -1.036963584997 -1.057620226004 -1.05734818144
0.02 -6.447607828038 6.445266535551 -0.002341292487189 -1.100631436507 -1.120215413372 -1.12
}\matchingdata
\begin{tikzpicture}
\begin{axis}[
name=upper,width=0.86\linewidth,height=5.0cm,
xmode=log,xmin=0.002,xmax=0.02,ymin=-7,ymax=7,
xtick={0.002,0.005,0.01,0.02},xticklabels={},
ylabel={Phase change (rad)},grid=major,
legend style={font=\small,at={(0.5,1.03)},anchor=south,legend columns=3,draw=none,fill=none},
tick label style={font=\small}]
\addplot[blue,thick] table[x=zeta,y=reflection] {\matchingdata};
\addlegendentry{$\Delta\phi_c$}
\addplot[red,dashed,thick] table[x=zeta,y=propagation] {\matchingdata};
\addlegendentry{$-2\omega\Delta x_0$}
\addplot[black,densely dotted,thick] table[x=zeta,y=total] {\matchingdata};
\addlegendentry{Sum}
\end{axis}
\begin{axis}[
at={(upper.south west)},anchor=north west,yshift=-0.85cm,
width=0.86\linewidth,height=5.0cm,
xmode=log,xmin=0.002,xmax=0.02,
xtick={0.002,0.005,0.01,0.02},xticklabels={0.002,0.005,0.01,0.02},
xlabel={$\zeta=d_0/(4M)$},ylabel={Residual phase (mrad)},
grid=major,tick label style={font=\small},
legend style={font=\small,at={(0.5,1.03)},anchor=south,legend columns=3,draw=none,fill=none}]
\addplot[blue,thick] table[x=zeta,y=smooth] {\matchingdata};
\addlegendentry{Smooth}
\addplot[red,dashed,thick] table[x=zeta,y=band] {\matchingdata};
\addlegendentry{Exact band}
\addplot[black,densely dotted,thick] table[x=zeta,y=geometry] {\matchingdata};
\addlegendentry{$-8M\omega\zeta^2$}
\end{axis}
\end{tikzpicture}
\caption{Matching-distance dependence at fixed proper-frame cutoff, with $M\omega=0.35$ and $M\Omega_{\rm UV}=M\Omega_b=10^4$.
Top: changes of the smooth outer phase and the propagation phase relative to $\zeta=0.002$, and their sum.
Bottom: $\Phi_{{\rm in},c}-\Phi_{0,c}$ and $\Phi_{{\rm in},b}-\Phi_{0,b}$, compared with the leading geometric term.
The band edges satisfy $\omega/\omega_{c,b}<0.018$ throughout.
This is a comparison within the specified response families, without a fitted delay or the exterior barrier phase.}
\label{fig:matching_scale}
\end{figure}

\subsection{Fixed-plane cutoff variation and delay alignment}

For fixed $\lambda>0$, as $\omega_c\to\infty$ at fixed positive $\omega$,
\begin{equation}
\phi_{\lambda c}(\omega)-\phi_c(\omega)
=\frac{\beta\omega}{\pi}\log\lambda
+O\left(\frac{\beta\omega^3}{\omega_c^2}
\log\frac{\omega_c}{\omega}\right).
\label{eq:UV_phase_scale_change}
\end{equation}
Here $\phi_{\lambda c}$ uses the cutoff $\lambda\omega_c$.
The modulus can therefore approach the Boltzmann law without convergence of the phase at a fixed plane.
To isolate the part not accounted for by the leading delay, define
\begin{eqnarray}
\Delta\tau_\lambda&=&\frac{\beta}{\pi}\log\lambda,
\label{eq:UV_leading_delay}\\
\epsilon_{\lambda,c}(\omega)&=&\phi_{\lambda c}(\omega)-\phi_c(\omega)
-\omega\Delta\tau_\lambda.
\label{eq:UV_residual_phase}
\end{eqnarray}
The $j=0$ terms in Eq.~(\ref{eq:UV_phase_series}) cancel exactly.
The $j=1$ terms give, at fixed $\lambda$ and positive $\omega$ as $\omega_c\to\infty$,
\begin{equation}
\begin{split}
\epsilon_{\lambda,c}(\omega)
={}&\frac{\beta\omega^3}{\pi\omega_c^2}
\Biggl\{(\lambda^{-2}-1)
\left[\frac{\log(\omega_c/\omega)}{3}+\frac19\right]
+\frac{\lambda^{-2}\log\lambda}{3}\Biggr\}\\
&+O\left(\frac{\beta\omega^5}{\omega_c^4}
\log\frac{\omega_c}{\omega}\right).
\end{split}
\label{eq:UV_residual_expansion}
\end{equation}
This residual is dispersive and suppressed by two powers of the cutoff, up to the logarithm.
Its form follows from the chosen smooth family and outer-response assumption; it is not a universal constraint on ultraviolet completions.
At the Schwarzschild temperature, $\Delta\tau_2=8M\log2$.
Proposition~\ref{prop:matching_scale_cancellation} shows why this same linear-delay coefficient appears in the near-horizon matching problem when the cutoff change is generated by gravitational redshift.
Section~\ref{sec:UV_interference} separates this leading delay from the fixed-plane cavity response, and Sec.~\ref{sec:UV_numerics} quantifies the residual without fitting additional phase parameters.

\section{Schwarzschild scattering and the passive thermal output}
\label{sec:composition}

\subsection{Passive horizon input}
\label{sec:quantum_completion}

For the application we use one effective passive environment input $c$, independent of the incident field of the isolated horizon component.
The standard attenuator relation is
\begin{equation}
a_+={\cal R}_H a_-+{\cal L}_H c,
\qquad |{\cal R}_H|^2+|{\cal L}_H|^2=1.
\label{eq:horizon_io}
\end{equation}
The outer factors constructed in Sec.~\ref{sec:causal_completion} supply causal ${\cal R}_H$ and ${\cal L}_H$.
Their modulus relation preserves the output commutator.
The state is a separate input: for a thermal port at inverse Killing temperature $\beta$, its occupation is $n_\beta(\omega)=(e^{\beta\omega}-1)^{-1}$.
The complete asymptotic scattering problem below uses independent external inputs; the internal cavity amplitudes are not independent canonical modes.

The noise commutator can equivalently be obtained from a Langevin boundary condition.
Appendix~\ref{app:response} gives that derivation and the equilibrium response kernels.
Additional annihilation and creation inputs, needed for more general channels, are described in Appendix~\ref{app:general_channel}.
An explicit causal unitary two-port realization is given in Appendix~\ref{app:rational}.

\subsection{Barrier composition and stability}

Let $\psi_{\rm up}$ be the solution of Eq.~(\ref{eq:rw_general}) that is purely outgoing at infinity,
\begin{equation}
\psi_{\rm up}(r_*)\sim e^{i\omega r_*},
\qquad
r_*\to+\infty.
\label{eq:up_infinity}
\end{equation}
Near the horizon side of the barrier,
\begin{equation}
\psi_{\rm up}(r_*)
\sim
B_{\rm out}(\omega)e^{i\omega r_*}
+B_{\rm in}(\omega)e^{-i\omega r_*}.
\label{eq:up_horizon}
\end{equation}
For incidence on the Regge--Wheeler barrier from the left, define \cite{Mark2017,Newton1982}
\begin{equation}
{\cal R}_{\bh}(\omega)
=\frac{B_{\rm in}}{B_{\rm out}},
\qquad
{\cal T}_{\bh}(\omega)
=\frac{1}{B_{\rm out}}.
\label{eq:BH_RT}
\end{equation}
For real $\omega$ and a real potential,
\begin{equation}
|{\cal R}_{\bh}|^2+|{\cal T}_{\bh}|^2=1.
\label{eq:BH_flux}
\end{equation}

Repeated propagation between the matching surface and the barrier gives the standard denominator
\begin{equation}
D(\omega)
=1-
{\cal R}_{\bh}(\omega)
{\cal R}_H(\omega)
e^{-2i\omega x_0}.
\label{eq:echo_denominator}
\end{equation}
For clarity, this denominator follows directly from the feedback equation
\begin{equation}
a_-=a_{\rm p}+{\cal R}_{\bh}e^{-2i\omega x_0}a_+,
\label{eq:feedback}
\end{equation}
where $a_{\rm p}$ is a prescribed inward driving amplitude at the matching surface before its first interaction with the horizon.
It is an inner source normalization, not by itself an independently accessible quantum port at infinity.
It can represent the horizon-going waveform used as the source term of an echo transfer function \cite{Mark2017,KaminagaEcho2026}.
Substitution of Eq.~(\ref{eq:feedback}) into Eq.~(\ref{eq:horizon_io}) gives the factor $D^{-1}$.
The part transmitted to infinity is ${\cal T}_{\bh}e^{-i\omega x_0}a_+$.
A physical asymptotic input and the prompt reflected wave are restored in Sec.~\ref{sec:complete_asymptotic}.
This fixes the phase convention used below; a change of reference plane alters phases but not any modulus relation.
It is useful to introduce the common exterior filter
\begin{equation}
{\cal F}(\omega)
=
\frac{{\cal T}_{\bh}(\omega)e^{-i\omega x_0}}
{D(\omega)}.
\label{eq:common_filter}
\end{equation}

The finite-plane barrier return is
\begin{equation}
\widehat{\cal R}_{\bh,x_0}(z)={\cal R}_{\bh}(z)e^{-2izx_0}.
\label{eq:delayed_barrier_return}
\end{equation}
For the exact matched exterior with nonnegative potential, the energy estimate in Appendix~\ref{app:stability} proves that this return is analytic and has modulus strictly smaller than one for $\ImPart z>0$.
Since an admissible ${\cal R}_H$ has modulus at most one there, their product cannot equal one.
Thus $D$ has no exponentially growing feedback pole.
This statement does not give a uniform real-axis bound on $D^{-1}$ when the loop modulus approaches one, for example at a threshold.

The inner-drive transfer and the horizon-port transfer are
\begin{equation}
{\cal K}={\cal F}{\cal R}_H,
\qquad {\cal L}_{\rm eff}={\cal F}{\cal L}_H.
\label{eq:K_transfer}
\end{equation}
Here $|{\cal K}|^2$ is a transfer strength in the prescribed-source normalization, not the reflection probability of an independent asymptotic input.
For nonzero ${\cal K}$, the common filter cancels in
\begin{equation}
\frac{|{\cal L}_{\rm eff}|^2}{|{\cal K}|^2}
=\frac{1-|{\cal R}_H|^2}{|{\cal R}_H|^2}.
\label{eq:central_ratio}
\end{equation}
The same denominator controls local resonance poles when the numerators do not vanish.
Appendix~\ref{app:poles} records this property and the first-order pole displacement.

\subsection{Complete asymptotic channel}
\label{sec:complete_asymptotic}

Introduce the two barrier-side amplitudes $u_+$ and $u_-$ and write the lossless two-port barrier relation as
\begin{equation}
\left(\begin{array}{c}
a_{\out}\\ u_-
\end{array}\right)
=
S_{\bh}
\left(\begin{array}{c}
a_{\rm in}\\ u_+
\end{array}\right),
\label{eq:barrier_two_port}
\end{equation}
where
\begin{equation}
S_{\bh}
=
\left(\begin{array}{cc}
{\cal R}_{\infty} & {\cal T}_{\bh}\\
{\cal T}_{\rm in}^{(0)} & {\cal R}_{\bh}
\end{array}\right)
\qquad
S_{\bh}^\dagger S_{\bh}=I.
\label{eq:barrier_matrix}
\end{equation}
Here $a_{\rm in}$ and $a_{\out}$ are the asymptotic modes in the left panel of Fig.~\ref{fig:channel_schematic}, while
\begin{equation}
u_+=e^{-i\omega x_0}a_+,
\qquad
u_-=e^{i\omega x_0}a_-.
\label{eq:barrier_surface_phases}
\end{equation}
Define the inward amplitude including its reference-plane phase by
\begin{equation}
{\cal T}_{\rm in}
=e^{-i\omega x_0}{\cal T}_{\rm in}^{(0)}.
\label{eq:Tin_phase_definition}
\end{equation}
The feedback equation then becomes
\begin{equation}
a_-
={\cal T}_{\rm in}a_{\rm in}
+{\cal R}_{\bh}e^{-2i\omega x_0}a_+.
\label{eq:complete_feedback}
\end{equation}

Solving the feedback equation together with Eq.~(\ref{eq:horizon_io}) gives
\begin{eqnarray}
a_{\out}&=&{\cal S}_{\rm tot}a_{\rm in}+{\cal L}_{\rm eff}c,
\label{eq:complete_asymptotic_map}\\
{\cal S}_{\rm tot}&=&{\cal R}_\infty+{\cal K}{\cal T}_{\rm in},
\label{eq:Stot}\\
{\cal K}_\infty&=&{\cal K}{\cal T}_{\rm in}.
\label{eq:asymptotic_input_transfer}
\end{eqnarray}
The first term in ${\cal S}_{\rm tot}$ is prompt reflection, and ${\cal K}_\infty$ is the echo correction for incidence from infinity.
For a reciprocal Schwarzschild barrier, $|{\cal T}_{\rm in}|=|{\cal T}_{\bh}|$.
Barrier unitarity gives
\begin{equation}
1-|{\cal S}_{\rm tot}|^2
=|{\cal F}|^2(1-|{\cal R}_H|^2)
=|{\cal L}_{\rm eff}|^2.
\label{eq:Stot_absorption_identity}
\end{equation}
This is the standard absorption-noise identity \cite{Beenakker1998}; the short algebraic derivation is included in Appendix~\ref{app:general_channel}.
It retains the interference between prompt reflection and echoes.

For vacuum incident from infinity and a passive thermal horizon input, the normally ordered horizon contribution is
\begin{equation}
N_{\out}^{H}=|{\cal L}_{\rm eff}|^2n_\beta
=\left(1-|{\cal S}_{\rm tot}|^2\right)n_\beta.
\label{eq:Nout}
\end{equation}
With an uncorrelated incident occupation $n_\infty$, one adds $|{\cal S}_{\rm tot}|^2n_\infty$.
In particular, $n_\infty=n_\beta$ gives the equilibrium output $n_\beta$.
Assigning a thermal horizon input and vacuum at infinity is an effective Unruh-type prescription \cite{Unruh1976}, not a claim that the full Unruh state is KMS.

For vacuum incidence from infinity, the complete symmetrized output variance is
\begin{equation}
{\cal S}_{\out}^{\rm full}=\frac12+N_{\out}^{H}.
\label{eq:complete_passive_covariance}
\end{equation}
Two vacuum inputs therefore give exactly the vacuum variance.
A vacuum term assigned to the horizon port alone is not an additional observable noise floor.
At perfect absorption, ${\cal R}_H\to0$, one has
\begin{equation}
N_{\out}^{H}\longrightarrow |{\cal T}_{\bh}|^2n_\beta.
\label{eq:greybody_thermal}
\end{equation}
At $\beta=\beta_H$ this is the usual greybody-filtered Planck form \cite{Hawking1975,Page1976}.
The temperature has been imposed on the input; the channel does not derive Hawking radiation or its partner entanglement.

\subsection{Thermal matching and its normalization}
\label{sec:boltzmann_matching}

Equations~(\ref{eq:central_ratio}) and (\ref{eq:Nout}) give
\begin{equation}
\frac{N_{\out}^{H}}{|{\cal K}|^2}
=\frac{1-|{\cal R}_H|^2}{|{\cal R}_H|^2}\,n_\beta.
\label{eq:Nout_consistency}
\end{equation}
If $|{\cal R}_H|^2=e^{-\beta\omega}$ in a band and the passive input is thermal at the same $\beta$, the right-hand side equals one:
\begin{equation}
N_{\out}^{H}(\omega)=|{\cal K}(\omega)|^2.
\label{eq:boltzmann_equality}
\end{equation}
The exact-band completion realizes this conditional identity throughout $0<\omega\leq\omega_b$.
For a coherent inner drive in a normalized narrow wave packet with amplitude $\alpha$, suppose that the transfer coefficients and the noise occupation vary negligibly across its bandwidth.
The echo occupation is then approximately $|{\cal K}|^2|\alpha|^2$, and the noise-to-coherent ratio is $1/|\alpha|^2$ to the same accuracy.
This is consistent with omitting stochastic fluctuations in a large-amplitude calculation \cite{OshitaWangAfshordi2020}.

For incidence from infinity, the echo correction is instead ${\cal K}_\infty$, and
\begin{equation}
\frac{N_{\out}^{H}}{|{\cal K}_\infty|^2}
=\frac{1}{|{\cal T}_{\rm in}|^2}
\frac{1-|{\cal R}_H|^2}{|{\cal R}_H|^2}\,n_\beta.
\label{eq:asymptotic_ratio}
\end{equation}
These are stationary spectral relations.
A time window also acts on the noise covariance and cannot be inserted without bandwidth and signal-separation assumptions.
Active inputs require the occupation formula in Appendix~\ref{app:general_channel}, not Eq.~(\ref{eq:boltzmann_equality}).

For the smooth family in Eq.~(\ref{eq:UV_flux_reflectivity}) and a KMS input at inverse temperature $\beta$, the equality in Eq.~(\ref{eq:boltzmann_equality}) is replaced by
\begin{equation}
\frac{N_{\out}^{H}(\omega)}{|{\cal K}(\omega)|^2}
=\chi_{\rm UV}(\omega)
:=\frac{W_{\rm UV}(\omega)^{-1}-1}
{e^{\beta\omega}-1}.
\label{eq:UV_matching_factor}
\end{equation}
For $\omega>0$, define
\begin{equation}
\Delta_{\rm UV}(\omega)
=\beta\left[\omega-\omega_c\log(1+\omega/\omega_c)\right].
\label{eq:UV_log_error}
\end{equation}
The inequalities $x-x^2/2\leq\log(1+x)\leq x$ for $x\geq0$ give $0\leq\Delta_{\rm UV}\leq\beta\omega^2/(2\omega_c)$.
Therefore
\begin{equation}
0<\chi_{\rm UV}\leq1,
\qquad
1-\chi_{\rm UV}
=\frac{1-e^{-\Delta_{\rm UV}}}{1-e^{-\beta\omega}}
\leq\frac{\beta\omega^2}{2\omega_c(1-e^{-\beta\omega})}.
\label{eq:UV_error_bound}
\end{equation}
In the regime $\beta\omega\gtrsim1$, a small fractional correction requires $\beta\omega^2/\omega_c$ to be small; the condition $\omega/\omega_c\ll1$ alone does not control it.
This matters when $\beta\omega_c$ is large.
This bound is specific to the smooth family; it is not a lower bound on the correction of an arbitrary causal completion.
The ratio is phase independent, although the absolute spectra are not.

At fixed frequency, $\chi_{\rm UV}\to1$ as $\omega_c\to\infty$.
Figure~\ref{fig:boltzmann_uv} shows the finite-band correction, while Eq.~(\ref{eq:UV_phase_scale_change}) explains why this convergence does not extend to the phase.

\begin{figure}[htbp]
\centering
\includegraphics[width=0.82\linewidth]{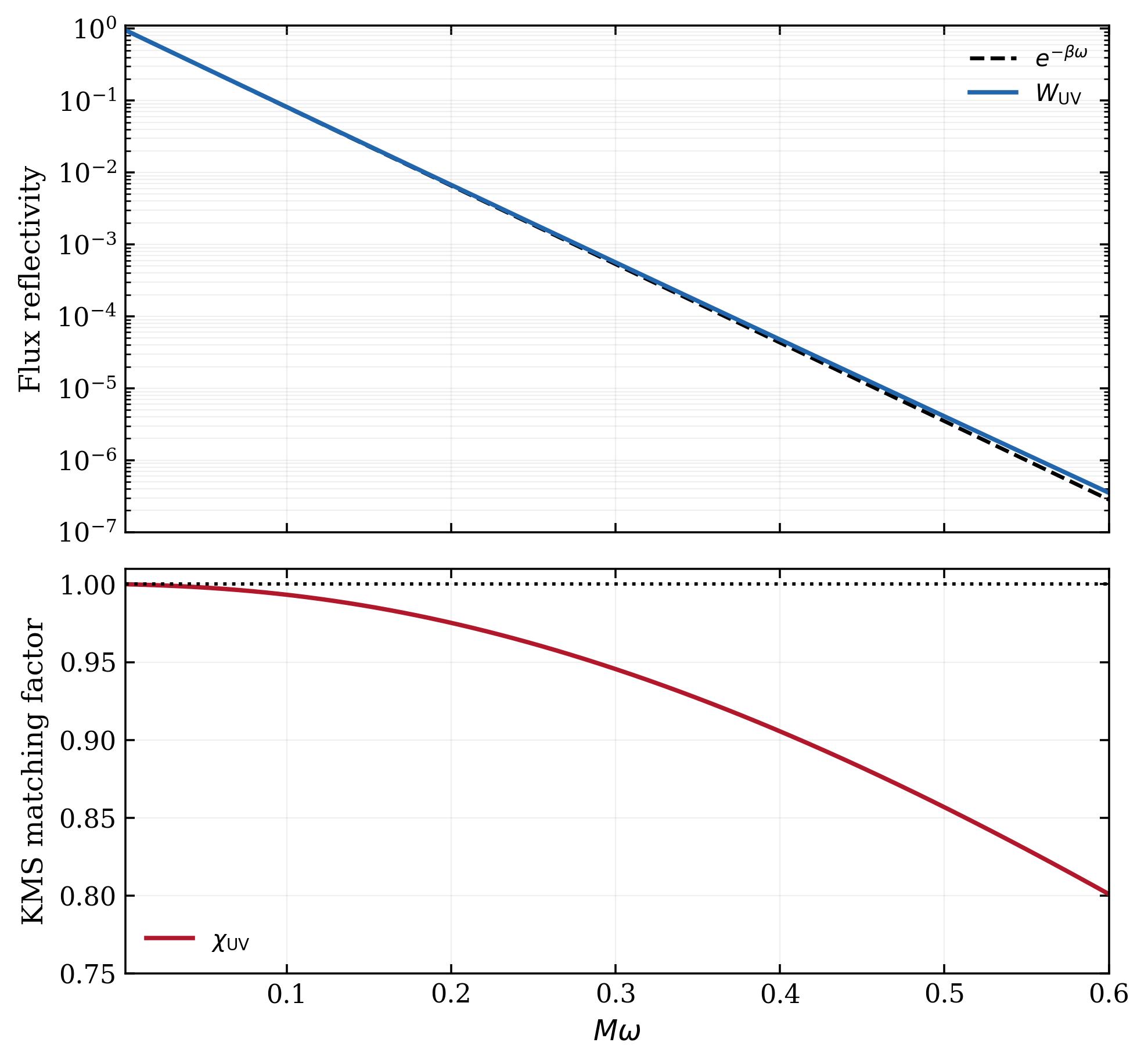}
\caption{Causal ultraviolet completion of the Boltzmann flux law for $\beta=8\pi M$ and $M\omega_c=20$.
The upper panel compares $W_{\rm UV}$ in Eq.~(\ref{eq:UV_flux_reflectivity}) with $e^{-\beta\omega}$.
The lower panel shows the correction factor $\chi_{\rm UV}$ in Eq.~(\ref{eq:UV_matching_factor}).
The completion tracks the exponential law in the lower part of the effective band, while the relative departure increases across the plotted band according to Eq.~(\ref{eq:UV_error_bound}); the crossover $M\omega_c=20$ is outside the plot.
The ultraviolet tail remains compatible with a bounded causal outer function.
The plot was produced using AI-assisted Python code described in Appendix~\ref{app:numerical}.}
\label{fig:boltzmann_uv}
\end{figure}

The absolute scale should also be kept separate from the normalized equality.
For a Schwarzschild black hole, $\beta_H=8\pi M$, and hence at the representative frequency $M\omega=0.35$ in the axial ringdown band the Boltzmann flux reflectivity is approximately
\begin{equation}
e^{-8\pi M\omega}
\simeq1.5\times10^{-4}.
\label{eq:ringdown_boltzmann_scale}
\end{equation}
At lower frequencies the reflectivity can be appreciable, but the $s=2$, $\ell=2$ barrier transmission is then strongly suppressed.
Thus Eqs.~(\ref{eq:boltzmann_equality}) and (\ref{eq:UV_matching_factor}) are primarily consistency relations; they do not by themselves imply an observable quantum-noise signal.

\subsection{Axial barrier and the outgoing spectrum}
\label{sec:numerics}

We use the axial $s=2$, $\ell=2$ potential with $M=1$.
The outgoing solution is initialized at $r_*=140$ and integrated inward to $r_*=-50$.
The coefficients in Eq.~(\ref{eq:BH_RT}) are extracted from the free-wave decomposition.
On the frequency grid the flux defect is below $4\times10^{-8}$; this tests conservation, not the error from truncating the infinite domain.
Outer-boundary changes give relative variations below $0.73\%$ for the positive transfer strengths checked in Appendix~\ref{app:numerical}.
The illustrative spectra below retain the original finite-domain calculation.
Appendix~\ref{app:rational} gives separate channel checks for a rational response.

We apply the smooth completion to this numerical barrier, using the outer phase in Eqs.~(\ref{eq:outer_phase_quadrature}) and (\ref{eq:UV_phase_integral}).
Thus
\begin{equation}
{\cal R}_{H,{\rm UV}}(\omega)=\sqrt{W_{\rm UV}(\omega)}e^{i\phi_c(\omega)}
\label{eq:numerical_UV_reflection}
\end{equation}
and define
\begin{equation}
D_{\rm UV}(\omega)
:=1-{\cal R}_{\bh}(\omega){\cal R}_{H,{\rm UV}}(\omega)
e^{-2i\omega x_0}.
\label{eq:numerical_UV_denominator}
\end{equation}
With one KMS loss input, the resulting horizon-sourced occupation at infinity is
\begin{equation}
N_{\out}^{H}(\omega)
=\frac{|{\cal T}_{\bh}|^2
[1-W_{\rm UV}(\omega)]}{|D_{\rm UV}(\omega)|^2}
\frac{1}{e^{\beta\omega}-1}.
\label{eq:UV_asymptotic_noise_spectrum}
\end{equation}
This is the spectral occupation at infinity for the selected phase and state.
Its conversion to energy flux is given by Eq.~(\ref{eq:stationary_energy_flux}).
It is more information than the ratio in Eq.~(\ref{eq:UV_matching_factor}), but depends on the chosen ultraviolet completion and phase.

One can also bound this spectrum without selecting the horizon phase.
Put $a(\omega)=|{\cal R}_{\bh}(\omega)|\sqrt{W_{\rm UV}(\omega)}$ and $A(\omega)=|{\cal T}_{\bh}(\omega)|^2[1-W_{\rm UV}(\omega)]n_\beta(\omega)$.
For $\omega>0$, $a<1$, and the triangle inequality gives
\begin{equation}
\frac{A(\omega)}{[1+a(\omega)]^2}
\leq N_{\out}^{H}(\omega)
\leq\frac{A(\omega)}{[1-a(\omega)]^2}.
\label{eq:phase_independent_noise_bounds}
\end{equation}
The bounds hold for every admissible phase with the stated modulus.
They do not assert that both bounds can be attained at all frequencies by one causal response.
When the horizon reflectivity is small, they show directly that the noise stays close to the perfectly absorbing spectrum even if the echo phase is uncertain.

\begin{figure}[htbp]
\centering
\includegraphics[width=0.82\linewidth]{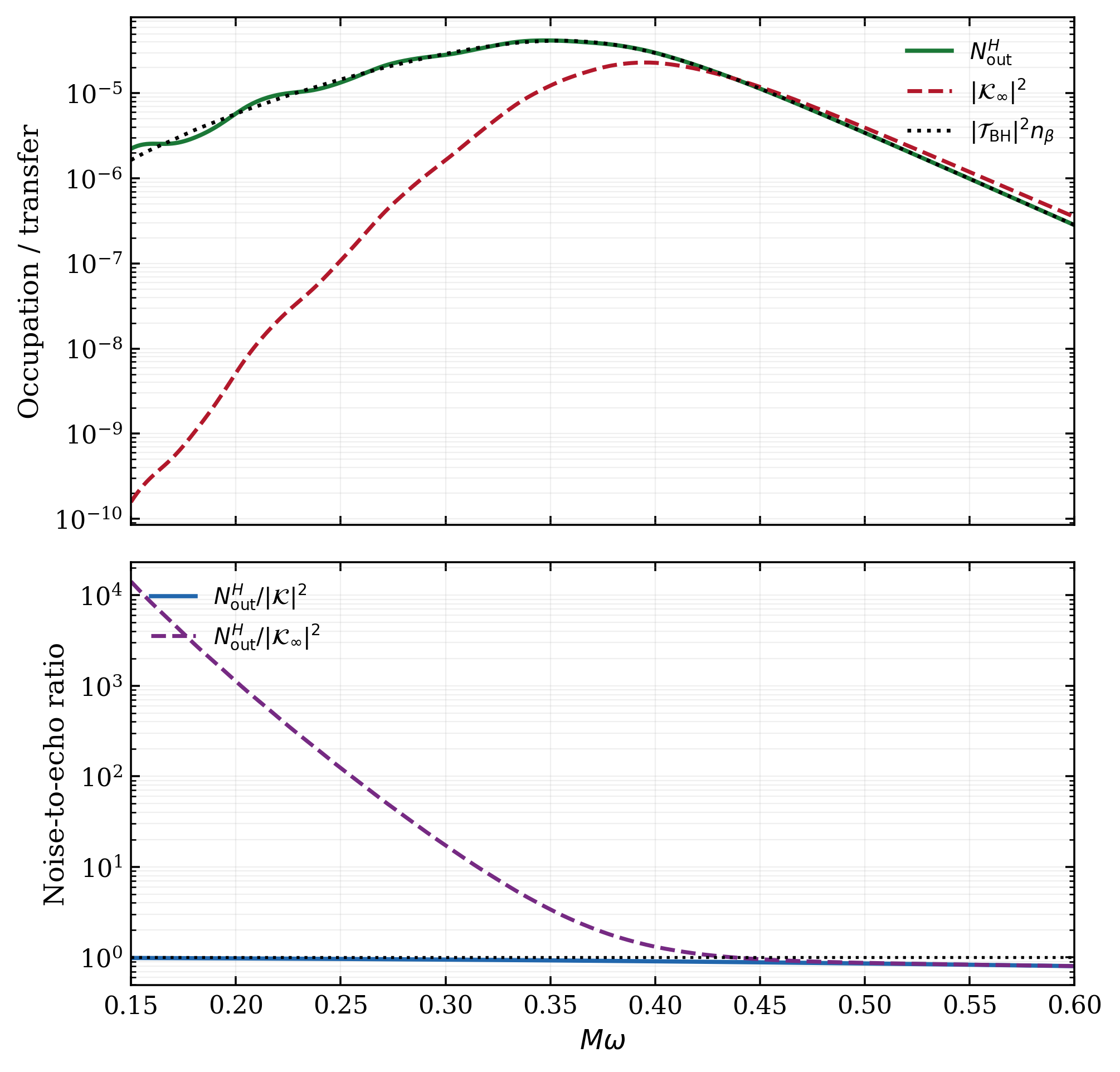}
\caption{Ultraviolet-completed quantum-noise spectrum for the axial $s=2$, $\ell=2$ barrier, with $\beta=8\pi M$, $M\omega_c=20$, and $x_0/M=-30$.
The upper panel compares the horizon-sourced occupation in Eq.~(\ref{eq:UV_asymptotic_noise_spectrum}), the coherent echo transfer $|{\cal K}_\infty|^2$ for a unit occupation incident from infinity, and the perfectly absorbing result $|{\cal T}_{\bh}|^2n_\beta$.
The lower panel shows the horizon-normalized ratio $N_{\out}^{H}/|{\cal K}|^2=\chi_{\rm UV}$ and the spectral noise-to-echo ratio $N_{\out}^{H}/|{\cal K}_\infty|^2=\chi_{\rm UV}/|{\cal T}_{\rm in}|^2$.
The large low-frequency value of the latter is the inward greybody suppression, not a cavity enhancement of the local identity.
The plot was produced using AI-assisted Python code described in Appendix~\ref{app:numerical}.}
\label{fig:UV_quantum_noise}
\end{figure}

\subsection{Fixed-plane interference and delay alignment}
\label{sec:UV_interference}

Proposition~\ref{prop:matching_scale_cancellation} compares the specified response family at different matching surfaces and fixed local cutoff, for which the redshifted cutoff and the propagation length change together.
It is also useful to vary $\omega_c$ at fixed $x_0$ in order to isolate the response-side phase before this compensation.
Let $D_{\rm UV}$ be Eq.~(\ref{eq:echo_denominator}) with the smooth outer reflection.
On a positive real-frequency interval with ${\cal R}_{\bh}\neq0$, choose a continuous barrier phase and write
\begin{eqnarray}
a_c(\omega)&=&|{\cal R}_{\bh}(\omega)|\sqrt{W_{\rm UV}(\omega)},\\
\Theta_c(\omega)&=&\arg{\cal R}_{\bh}(\omega)+\phi_c(\omega)+2L\omega,\qquad L=-x_0,\nonumber\\
{\cal C}_c(\omega):=|D_{\rm UV}(\omega)|^{-2}
&=&[1+a_c^2-2a_c\cos\Theta_c]^{-1}.
\label{eq:UV_cavity_modulation}
\end{eqnarray}
The locations $\omega_n^{\rm ph}$ satisfying $\Theta_c(\omega_n^{\rm ph})=2\pi n$ identify constructive phase alignment.
They need not coincide with maxima of ${\cal C}_c$ when $a_c$ varies, and they are not complex quasinormal frequencies.
The derivative of a simple phase crossing is given in Appendix~\ref{app:poles}.
The following comparison retains both the phase and the loop modulus.

For $\lambda>1$, the leading delay can be aligned by adding a positive delay to the lower-cutoff response:
\begin{eqnarray}
{\cal R}^{(d)}_{H,c}(z)&=&e^{iz\Delta\tau_\lambda}{\cal R}_{H,c}(z),
\label{eq:UV_delayed_response}\\
{\cal C}^{(d)}_c(\omega)&=&
[1+a_c^2-2a_c\cos(\Theta_c+\omega\Delta\tau_\lambda)]^{-1}.
\label{eq:UV_delayed_cavity}
\end{eqnarray}
The factor $e^{iz\Delta\tau_\lambda}$ is bounded by unity in $\mathbb C_+$ and has unit real-axis modulus, so this operation preserves causal passivity.
The barrier and reference plane remain fixed.
In the round-trip denominator it is equivalent to increasing $L$ by $\Delta\tau_\lambda/2$.
We compare this delayed response with ${\cal R}_{H,\lambda c}$: their horizon phase difference is exactly $\epsilon_{\lambda,c}$, while their different moduli are retained.
This fixed-plane delay alignment is a diagnostic; it is not itself a displacement of the matching surface.
Under the local-cutoff interpretation of Sec.~\ref{sec:redshift_local_cutoff}, an outward shift of the matching surface raises $\omega_c$ but decreases $L$, and the two leading phase changes have the opposite signs given in Proposition~\ref{prop:matching_scale_cancellation}.
The numerical alignment below measures how much of the fixed-plane cutoff dependence is carried by this linear-delay component.
The remaining dispersion and modulus differences must be quantified before assigning physical significance to an ultraviolet interference shift.

\subsection{Numerical cutoff comparison}
\label{sec:UV_numerics}

We compare $M\omega_c=20$ and $40$ at $\beta=8\pi M$ and $x_0/M=-30$, with the same outer-response assumption.
This is deliberately a fixed-plane comparison: it isolates the response-side cutoff dependence before the matching-scale compensation of Proposition~\ref{prop:matching_scale_cancellation}.
In the local interpretation of Eq.~(\ref{eq:cutoff_proper_distance}), the two cutoffs correspond at leading near-horizon order to $d_0\Omega_{\rm UV}\simeq80$ and $160$, respectively.
These values satisfy the hierarchy in Eq.~(\ref{eq:physical_cutoff_hierarchy}) throughout the plotted band.
At $M\omega=0.35$, Table~\ref{tab:UV_phase} separates the inner thermal ratio from the horizon phase and its group delay.
Although the ratio changes by about $3.6$ percentage points, the phase changes by $1.940$ radians.
The horizon contribution to the group delay changes by $5.538M$, close to $8M\log2=5.545M$.
This is only the response part of the delay; a correlated change of the matching surface within the specified local-response family changes the propagation term with the compensating leading sign.

\begin{table}[htbp]
\caption{Smooth-completion data at $M\omega=0.35$.
The phase is continuous from $\phi_c(0^+)=0$, without an added inner factor.}
\label{tab:UV_phase}
\centering
\begin{tabular}{cccc}
\toprule
$M\omega_c$ & $\chi_{\rm UV}$ & $\phi_c$ (rad) & $\tau_H/M$\\
\midrule
20 & 0.92673 & 14.12880 & 32.37435\\
40 & 0.96246 & 16.06873 & 37.91252\\
\bottomrule
\end{tabular}

\end{table}

At the same frequency, subtracting $\omega\Delta\tau_2$ leaves only
$\epsilon_{2,c}=-8.8940\times10^{-4}$ rad for $M\omega_c=20$.
Over $0.15\leq M\omega\leq0.60$, its largest absolute value is
$3.9002\times10^{-3}$ rad, attained at the upper endpoint.
The leading dispersive term in Eq.~(\ref{eq:UV_residual_expansion}) approximates the exact residual to within $2.7\times10^{-6}$ rad on this band.
Thus the large unaligned phase difference mainly measures the leading delay, not a comparably large residual dispersion.

Figure~\ref{fig:UV_phase_cavity} compares the cavity modulation before and after adding $\Delta\tau_2=5.54518M$ to the lower-cutoff response.
This is equivalent in the denominator to $\Delta L=2.77259M$.
On the 301-point grid spanning the band above, the maximum of
$|{\cal C}_{40/M}/{\cal C}_{20/M}-1|$ is approximately $0.2100$, whereas the maximum of
$|{\cal C}_{40/M}/{\cal C}^{(d)}_{20/M}-1|$ is approximately $1.0940\times10^{-3}$.
The reduction from about $21\%$ to $0.11\%$ includes the modulus changes as well as the residual phase; it is not a phase-only comparison.
No delay or phase has been optimized to minimize these differences.
The maximum delay-aligned relative difference changes by less than
$2.1\times10^{-9}$ in absolute value under the paired boundary and grid checks reported in Appendix G.
The near coincidence therefore provides a numerical measure of the small remainder after the linear-delay component has been removed, while Proposition~\ref{prop:matching_scale_cancellation} gives the separate leading cancellation when the matching surface and redshifted cutoff are varied together in the specified family.

For a common continuous branch with $\arg{\cal R}_{\bh}(0.15/M)$ chosen in $(-\pi,\pi]$, Table~\ref{tab:UV_crossings} separates the unaligned and delay-aligned shifts of constructive phase crossings.
The latter are only $1.90\times10^{-6}$ to $6.47\times10^{-6}$ in $M\omega$.
Paired boundary and grid checks for these smaller differences are given in Appendix~\ref{app:numerical}.
These frequencies describe real-axis phase alignment, not complex quasinormal poles.
The example demonstrates that a substantial unaligned shift need not imply strong discrimination between ultraviolet scales when an inner delay is free.
Establishing detectability would additionally require a source model, detector noise, and an inference calculation.

\begin{figure}[htbp]
\centering
\includegraphics[width=0.70\linewidth]{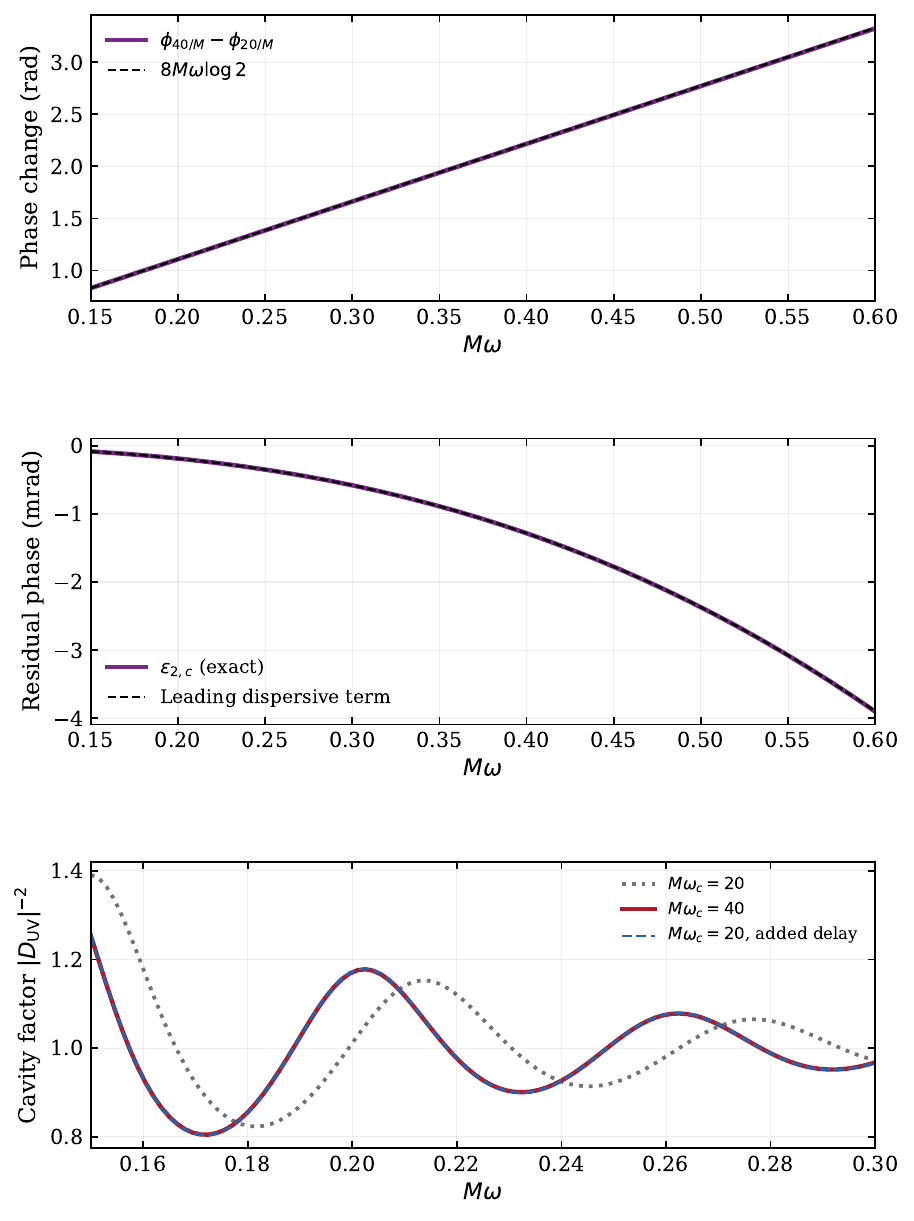}
\caption{Fixed-plane cutoff comparison at $x_0/M=-30$ and $\beta=8\pi M$.
Top: the outer-phase change from $M\omega_c=20$ to $40$ and its leading delay term $8M\omega\log2$.
Middle: the residual after subtracting this term, in milliradians, and the leading dispersive approximation in Eq.~(\ref{eq:UV_residual_expansion}).
Bottom: the original cavity factors for both cutoffs, together with the $M\omega_c=20$ response after adding the positive delay $\Delta\tau_2=8M\log2$.
The delayed blue curve nearly coincides with the red $M\omega_c=40$ curve.
Both moduli are retained.
This delay alignment is a fixed-plane diagnostic of the linear phase component; the matching-scale comparison additionally includes the redshift-related change of propagation length in Proposition~\ref{prop:matching_scale_cancellation}.
The bottom panel shows the lower part of the computed band; the relative differences quoted in the text use the full band.
The plot was produced using AI-assisted Python code described in Appendix~\ref{app:numerical}.}
\label{fig:UV_phase_cavity}
\end{figure}

\section{Discussion and conclusions}
\label{sec:conclusion}
\label{sec:limitations}

The two main conclusions concern causal realizability and the low-frequency phase of a class of ultraviolet completions.
An exact Boltzmann modulus on the entire positive frequency axis violates the logarithmic integrability condition for a nonzero bounded causal response at a fixed plane.
Exact finite-band matching remains possible.
This distinction separates the global transfer-function question from detailed balance or a low-energy reflection prescription.

Proposition~\ref{prop:general_outer_phase} identifies sufficient conditions under which a common Boltzmann slope gives the same logarithmic coefficient in the outer phase.
The completion enters the linear term through the explicit functional $C_h$.
At the Hawking temperature, redshifting a fixed local cutoff makes the logarithmic phase term cancel the leading matching-distance dependence of the tortoise propagation phase.
The geometric correction is quadratic in $d_0/M$.
A distinct dispersive correction depends on the completion, and its size relative to the geometric correction requires an additional hierarchy of scales.
The overlap requires $4M\omega\ll d_0\Omega_{\rm UV}$ as well as $d_0\ll M$.
It need not include the ringdown band if both the matching distance and the inverse local cutoff are of order the Planck length.
The strict event-horizon limit at fixed Killing frequency also leaves this regime.

This family comparison is not exact reference-plane invariance.
Exact transport of one physical boundary response must obey Eq.~(\ref{eq:reference_plane_transport}) or the wave-equation flow in Eq.~(\ref{eq:admittance_transport}).
The prescribed smooth family changes its modulus as well as its phase.
Moreover, an inner factor can add a delay not determined by the modulus.
The cancellation therefore applies to the stated class of outer responses; causality alone does not determine the complete echo delay.
Within these families, an echo delay cannot determine the matching distance without assumptions on the proper-frame cutoff and the inner phase.

The fixed-plane axial calculation gives a related diagnostic.
For $M\omega_c=20$ and $40$, removing the predicted linear-delay difference reduces the maximum relative cavity-factor difference from about $21\%$ to $0.11\%$ on the stated grid.
Both the modulus change and the residual dispersion are retained.
This illustrates a degeneracy between the cutoff-dependent phase and an inner delay; it is not a detector-level distinguishability calculation.

The passive quantum completion propagates horizon fluctuations through the same barrier and cavity and satisfies the standard absorption-noise identity.
Its thermal equality under finite-band Boltzmann matching is conditional on both passivity and the input temperature.
Neither it nor the perfect-absorption limit specifies the microscopic horizon state.
Linear response and two-point noise data also do not provide an injective reconstruction of horizon microstates.

The analysis assumes stationary, linear Schwarzschild scattering.
Kerr superradiance requires the Teukolsky flux normalization and an active-channel treatment; correlated or squeezed inputs require their full covariance, and nonstationary formation or evaporation requires two-time kernels.
A microscopic model would have to determine the response, state, cutoff, and any inner phase together.
The present constructions specify which global extensions are excluded and which part of the low-energy phase is common to completions satisfying the stated assumptions.

\section*{Acknowledgments}
OpenAI ChatGPT, including GPT-6 Sol and Astra Pro, was used for English editing, literature searches and bibliographic checks, and discussions and consistency checks of mathematical arguments. The Extra High thinking setting was used in some sessions. Other ChatGPT model versions may have been used during earlier stages of manuscript preparation, but their exact identifiers could not be reliably reconstructed. The author reviewed the AI-assisted material and takes full responsibility for the mathematical results, numerical calculations, and scientific conclusions. The use of AI-generated Python code for numerical calculations and figure preparation is described in Appendix~\ref{app:numerical}.

\section*{Data Availability}
No external data sets were used.
The script \path{generate_bhmap15_figures.py} reproduces Figs.~\ref{fig:boltzmann_uv}, \ref{fig:UV_quantum_noise}, \ref{fig:UV_phase_cavity}, and \ref{fig:coherent}, and exports the fixed-plane cutoff-comparison data.
The matching-scale plot in Fig.~\ref{fig:matching_scale} is embedded in the TeX source; its evaluation is described in Appendix~\ref{app:numerical}.
Figure~\ref{fig:channel_schematic} is supplied as a separate vector PDF.
Equations, parameters, integration boundaries, and numerical checks are given in the text and appendices.

\appendix
\numberwithin{equation}{section}
\renewcommand{\theequation}{\thesection\arabic{equation}}

\section{Canonical mode and flux normalization}
\label{app:normalization}

For each independent real axial harmonic with $\ell\geq2$, choose a canonically normalized master field satisfying the Regge--Wheeler equation.
We suppress the harmonic index below.
Its quadratic action has the form \cite{Moncrief1974}
\begin{equation}
S_{\ell m}^{(2)}
=\frac12\int dt\,dr_*
\left[
(\partial_t\Psi)^2-(\partial_{r_*}\Psi)^2
-V_{2\ell}\Psi^2
\right].
\label{eq:canonical_RW_action}
\end{equation}
Near a plane on which the potential is negligible, the corresponding operator has the expansion
\begin{equation}
\widehat\Psi(t,r_*)
=\int_0^\infty\frac{d\omega}{\sqrt{4\pi\omega}}
\left[
a_-(\omega)e^{-i\omega[t+(r_*-x_0)]}
+a_+(\omega)e^{-i\omega[t-(r_*-x_0)]}
+{\rm H.c.}
\right],
\label{eq:canonical_mode_expansion}
\end{equation}
where H.c. denotes the Hermitian-conjugate negative-frequency part required by
the Hermiticity of the real master field.
For an isolated scattering component, the incoming modes are normalized by
\begin{equation}
[a_-(\omega),a_-^\dagger(\omega')]
=\delta(\omega-\omega').
\label{eq:canonical_mode_commutator}
\end{equation}
The completed outgoing mode has the same self-commutator, but the incoming and outgoing modes are not independent.
For example, the isolated passive horizon relation gives $[a_-(\omega),a_+^\dagger(\omega')]={\cal R}_H(\omega')^*\delta(\omega-\omega')$.
After closing the cavity, $a_\pm$ denote internal field amplitudes determined by the external inputs.
Their commutators must then be calculated from that solution, rather than imposed as independent canonical commutators.
The normally ordered energy carried by one canonically normalized chiral wave packet is
\begin{equation}
{\cal E}
=\int_0^\infty d\omega\,
\omega\,\langle a^\dagger(\omega)a(\omega)\rangle.
\label{eq:canonical_energy_flux}
\end{equation}
For a finite excitation, Eq.~(\ref{eq:canonical_energy_flux}) is a wave-packet energy.
For a stationary state we instead define the spectral occupation by $\langle a^\dagger(\omega)a(\omega')\rangle=n(\omega)\delta(\omega-\omega')$.
The corresponding mean energy flux for one chiral mode is
\begin{equation}
P=\int_0^\infty\frac{d\omega}{2\pi}\,\omega n(\omega).
\label{eq:stationary_energy_flux}
\end{equation}
The factor $1/(2\pi)$ follows from the field normalization in Eq.~(\ref{eq:canonical_mode_expansion}), even with the delta function convention in Eq.~(\ref{eq:canonical_mode_commutator}).
The occupations in the text are spectral occupations in this sense.
For a normalized packet $a_f=\int_0^\infty f(\omega)^*a(\omega)d\omega$, its occupation is $\int|f|^2n\,d\omega$.
A frequencywise transfer formula applies to a packet without averaging only when its coefficients are approximately constant over the packet bandwidth.
Conversion from the canonical field to a metric waveform does not change the dimensionless transfer ratios.

\section{Boundary response and noise kernels}
\label{app:response}
\label{app:influence}

\subsection{Retarded self-energy}

A convenient deformation is
\begin{equation}
\Lambda(\omega)
=-i\omega+\Sigma_R(\omega),
\label{eq:Lambda_Sigma}
\end{equation}
so that
\begin{equation}
{\cal R}_H(\omega)
=
\frac{\Sigma_R(\omega)}
{2i\omega-\Sigma_R(\omega)}.
\label{eq:RH_Sigma}
\end{equation}
Here $\Sigma_R$ can be viewed as a matched retarded self-energy generated by unresolved near-horizon dynamics.

The self-energy convention of Eq.~(\ref{eq:Lambda_Sigma}) gives
\begin{equation}
\Sigma_R(z)
=\frac{2iz{\cal R}_H(z)}{1+{\cal R}_H(z)}.
\label{eq:Sigma_from_RH}
\end{equation}
This direction of construction is useful because stability is imposed on the measurable reflection amplitude itself.
In the opposite direction, analyticity of $\Sigma_R$ alone is insufficient: one must also require
\begin{equation}
2iz-\Sigma_R(z)\neq0,
\qquad
\ImPart z>0,
\label{eq:no_UHP_denominator_zero}
\end{equation}
in order that Eq.~(\ref{eq:RH_Sigma}) define a stable causal response.

For the following dispersion relation, assume that the retarded self-energy has a Fourier--Laplace transform analytic throughout $\ImPart z>0$.
Retarded support alone is not sufficient for this domain if an exponentially growing response is allowed.
If subtracting local terms leaves a remainder tending to a real constant $\Sigma_\infty$ and a convergent integral, its real and imaginary parts obey
\begin{equation}
\RePart\Sigma_R(\omega)-\Sigma_\infty
=
\frac{1}{\pi}\PV
\int_{-\infty}^{\infty}
\frac{\ImPart\Sigma_R(\omega')}
{\omega'-\omega}\dd\omega'.
\label{eq:KK}
\end{equation}
If these convergence conditions fail, further subtractions are needed; Eq.~(\ref{eq:KK}) is not an unsubtracted formula for an arbitrary growing self-energy.
Thus a response-generated reflectivity is not an arbitrary complex function: dissipation and phase are linked by causality, while stability adds the independent no-pole requirement in Eq.~(\ref{eq:no_UHP_denominator_zero}).

\subsection{Langevin boundary condition and the noise commutator}

The relation between boundary dissipation and the additional input operator can be made explicit.
Stochastic conditions of this type arise by integrating out degrees of freedom between a horizon and a stretched horizon \cite{SonTeaney2009,IsoOkazawa2011,CaronHuotCheslerTeaney2011}.
Here we use the corresponding frequency-domain operator condition for the canonically normalized mode,
\begin{equation}
\widehat\Psi'(x_0,\omega)
=\Lambda(\omega)\widehat\Psi(x_0,\omega)
+\Xi_{\rm phys}(\omega),
\label{eq:langevin_boundary}
\end{equation}
where $\Xi_{\rm phys}$ is an operator-valued boundary force in the field normalization of Eq.~(\ref{eq:canonical_RW_action}).
The positive-frequency part of Eq.~(\ref{eq:canonical_mode_expansion}) gives
\begin{equation}
(i\omega-\Lambda)a_+
=(i\omega+\Lambda)a_-
+\sqrt{4\pi\omega}\,\Xi_{\rm phys}.
\label{eq:physical_to_canonical_force}
\end{equation}
It is therefore useful to define the canonically mode-normalized input
\begin{equation}
\Xi_c(\omega)
:=\sqrt{4\pi\omega}\,\Xi_{\rm phys}(\omega).
\label{eq:canonical_boundary_force}
\end{equation}
Equation~(\ref{eq:physical_to_canonical_force}) then gives
\begin{equation}
a_+
={\cal R}_H a_-
+\frac{\Xi_c}{i\omega-\Lambda}.
\label{eq:langevin_io}
\end{equation}
For a passive set of independent input modes, write
\begin{equation}
\Xi_c(\omega)
=\sum_\alpha g_\alpha(\omega)c_\alpha(\omega),
\qquad
{\cal L}_{H\alpha}(\omega)
=\frac{g_\alpha(\omega)}{i\omega-\Lambda(\omega)}.
\label{eq:Xi_decomposition}
\end{equation}
Canonical commutator preservation and Eq.~(\ref{eq:passivity_identity}) then give the boundary fluctuation sum rule
\begin{equation}
\sum_\alpha|g_\alpha(\omega)|^2
=-4\omega\ImPart\Lambda(\omega),
\qquad
\omega>0.
\label{eq:boundary_noise_sum_rule}
\end{equation}
The right-hand side is nonnegative for a passive boundary and reproduces Eq.~(\ref{eq:horizon_io}).
The force normalization in Eq.~(\ref{eq:canonical_boundary_force}) relates this single-mode identity to the physical fluctuation-dissipation kernels: absorption and the force commutator are determined by the same admittance.

The sum rule fixes the real-frequency norm of the noise coupling, but not its phase or analytic continuation.
In the complex resonance statements of Appendix~\ref{app:poles}, we therefore assume that the local coefficients ${\cal L}_{H\alpha}(z)$ and ${\cal M}_{H\mu}(z)$ are stable causal transfer amplitudes, analytic in the upper half-plane and meromorphic in the continuation domain under consideration.
The rational model in Appendix~\ref{app:rational} gives an explicit analytic spectral factor rather than imposing only the modulus relation.

\subsection{Response and noise kernels}

A microscopic or effective hidden sector can be described by a Hilbert space $\HH_{\inr}$, Hamiltonian $H_{\inr}$, stationary state $\sigma$, and a Hermitian operator $O=O^\dagger$ coupled to an exterior amplitude.
We subtract its mean, so that $\Tr(\sigma O)=0$.
The two connected quadratic kernels are
\begin{equation}
G_R(t)
=-i\theta(t)
\Tr\left(\sigma[O(t),O(0)]\right)
\label{eq:GR_def}
\end{equation}
and
\begin{equation}
G_H(t)
=\frac12
\Tr\left(\sigma\{O(t),O(0)\}\right).
\label{eq:GH_def}
\end{equation}
The retarded kernel controls causal response, while the Hadamard kernel controls symmetrized fluctuations.
In a microscopic matching, $G_R$ contributes to $\Sigma_R$ or $\Lambda$, while $G_H$ fixes the covariance of $\Xi_{\rm phys}$.
Their normalizations depend on the coupling convention; after the conversion in Eq.~(\ref{eq:canonical_boundary_force}), the commutator part must reduce to Eq.~(\ref{eq:boundary_noise_sum_rule}).
This separation is standard in influence-functional and open-system descriptions \cite{FeynmanVernon1963,Kubo1957,HuPazZhang1992}; in black-hole effective theory closely related structures appear in Refs.~\cite{GoldbergerRothstein2006,GoldbergerRothstein2020,CaronHuotCheslerTeaney2011}.

With the Fourier transform $G(\omega)=\int_{-\infty}^{\infty}dt\,e^{i\omega t}G(t)$, the equilibrium KMS condition gives
\begin{equation}
G_H(\omega)
=-\coth\left(\frac{\beta\omega}{2}\right)
\ImPart G_R(\omega)
\label{eq:FDT}
\end{equation}
with the present Fourier convention and normalization.
Equivalently, for a bosonic input mode in a thermal state,
\begin{equation}
\left\langle c_\alpha^\dagger(\omega)c_\alpha(\omega')\right\rangle
=n_\alpha(\omega)\delta(\omega-\omega'),
\label{eq:thermal_occupation}
\end{equation}
where
\begin{equation}
n_\alpha(\omega)
=\frac{1}{e^{\beta\omega}-1}
\label{eq:planck_n}
\end{equation}
for a KMS bath with respect to the same Killing time.
The symmetrized single-mode variance is then
\begin{equation}
n_\alpha(\omega)+\frac12
=\frac12\coth\left(\frac{\beta\omega}{2}\right).
\label{eq:sym_kms}
\end{equation}

We emphasize that the commutator relation is a unitarity or commutator-preservation statement, whereas Eq.~(\ref{eq:FDT}) is a state-dependent equilibrium relation.
They play different roles and should not be conflated.

\subsection{Influence action}

For completeness, consider a coupling
\begin{equation}
H_{\rm coup}(t)=-\lambda q(t)O
\end{equation}
between an exterior amplitude and a hidden operator.
On the closed time path, define
\begin{equation}
q_c=\frac{q_++q_-}{2},
\qquad
q_\Delta=q_+-q_-.
\end{equation}
To quadratic order, integrating out the hidden sector gives the standard influence action
\begin{eqnarray}
S_{\rm IF}^{(2)}
&=&
-\lambda^2\int dt\,dt'\,
q_\Delta(t)G_R(t-t')q_c(t')
\nonumber\\
&&
+\frac{i\lambda^2}{2}\int dt\,dt'\,
q_\Delta(t)G_H(t-t')q_\Delta(t').
\label{eq:influence_action}
\end{eqnarray}
The real causal term modifies the coherent equation of motion, while the imaginary term generates fluctuations and decoherence.
This standard separation explains why a response-generated coherent reflectivity and a quantum noise kernel are distinct data before equilibrium conditions are imposed.

\section{General inputs and asymptotic commutators}
\label{app:general_channel}

\subsection{General linear input-output relation}

We first consider the isolated horizon component with independent incident and hidden input fields.
All frequencywise equations below are exact in the spectral representation after removing a common $\delta(\omega-\omega')$.
Narrow packets give the corresponding approximation described after Eq.~(\ref{eq:stationary_energy_flux}).
A stationary linear completion in the signed-frequency sector considered here has the form
\begin{equation}
a_+(\omega)
={\cal R}_H(\omega)a_-(\omega)
+\sum_\alpha {\cal L}_{H\alpha}(\omega)c_\alpha(\omega)
+\sum_\mu {\cal M}_{H\mu}(\omega)\widetilde d_\mu^\dagger(\omega).
\label{eq:horizon_bogoliubov}
\end{equation}
The $c_\alpha$ are annihilation-type hidden inputs.
In a signed-frequency representation the partner term is $d_\mu^\dagger(-\omega)$; we denote the same positive-norm partner operator by
\begin{equation}
\widetilde d_\mu^\dagger(\omega)
\equiv d_\mu^\dagger(-\omega),
\qquad \omega>0.
\label{eq:partner_frequency_definition}
\end{equation}
Equivalently, one may label it by a positive partner energy with reversed Killing-time orientation.
This convention makes the opposite signed Killing frequency explicit without treating a negative-norm mode as an independent positive-frequency annihilator.
The $\widetilde d_\mu^\dagger$ terms allow an active or Bogoliubov component.
Such creation-operator mixing is required in phase-insensitive linear amplification and is natural in quantum-field-theoretic black-hole channels \cite{Caves1982,BradlerAdami2015}.

For this isolated component, assume canonical bosonic commutators and vanishing cross commutators among its incident and hidden inputs.
The formula does not include a general frequency-mixing or phase-sensitive transformation of the exterior input.
Preservation of the output commutator gives
\begin{equation}
|{\cal R}_H|^2
+\sum_\alpha|{\cal L}_{H\alpha}|^2
-\sum_\mu|{\cal M}_{H\mu}|^2
=1.
\label{eq:horizon_commutator}
\end{equation}

The minus sign is the canonical signature of the creation-operator inputs.
Equation~(\ref{eq:horizon_commutator}) is local and does not assume thermal equilibrium.

The passive multi-input case sets all ${\cal M}_{H\mu}$ to zero, giving $\sum_\alpha|{\cal L}_{H\alpha}|^2=1-|{\cal R}_H|^2$.
In the cavity, ${\cal L}_{{\rm eff},\alpha}={\cal F}{\cal L}_{H\alpha}$ and
\begin{equation}
{\cal M}_{{\rm eff},\mu}={\cal F}{\cal M}_{H\mu}.
\label{eq:Meff}
\end{equation}
For nonzero ${\cal K}$, the common filter therefore gives
\begin{equation}
\frac{\sum_\alpha|{\cal L}_{{\rm eff},\alpha}|^2-\sum_\mu|{\cal M}_{{\rm eff},\mu}|^2}{|{\cal K}|^2}
=\frac{1-|{\cal R}_H|^2}{|{\cal R}_H|^2}.
\label{eq:bog_central_ratio}
\end{equation}
This is an inner-plane identity, not a barrier-independent ratio for an incident quantum port at infinity.

\subsection{Full output and its covariance}

The general asymptotic map is
\begin{equation}
a_{\out}={\cal S}_{\rm tot}a_{\rm in}
+\sum_\alpha{\cal L}_{{\rm eff},\alpha}c_\alpha
+\sum_\mu{\cal M}_{{\rm eff},\mu}\widetilde d_\mu^\dagger.
\label{eq:general_asymptotic_map}
\end{equation}
It obeys the standard full-system identity \cite{Beenakker1998}
\begin{equation}
|{\cal S}_{\rm tot}|^2+\sum_\alpha|{\cal L}_{{\rm eff},\alpha}|^2
-\sum_\mu|{\cal M}_{{\rm eff},\mu}|^2=1.
\label{eq:asymptotic_commutator}
\end{equation}
To verify this, set $r={\cal R}_{\bh}$, $q={\cal R}_H e^{-2i\omega x_0}$, and $\Delta_S=\det S_{\bh}$.
Barrier unitarity gives $|\Delta_S|=1$, ${\cal R}_\infty=\Delta_S r^*$, and $|{\cal T}_{\bh}|^2=1-|r|^2$, so
\begin{equation}
{\cal S}_{\rm tot}=\Delta_S\frac{r^*-q}{1-rq}.
\label{eq:Stot_unitary_form}
\end{equation}
Subtracting its modulus squared from one proves Eq.~(\ref{eq:Stot_absorption_identity}).
Combining that identity with Eq.~(\ref{eq:horizon_commutator}) proves Eq.~(\ref{eq:asymptotic_commutator}).
Only the independent external inputs enter this calculation.

For a general Bogoliubov completion, the normally ordered horizon contribution depends on the state as well as the commutator relation.
For independent stationary hidden inputs with zero means and no correlations between distinct channels,
\begin{equation}
N_{\out}^{H}
=|{\cal F}|^2
\left[
\sum_\alpha|{\cal L}_{H\alpha}|^2 n_{c,\alpha}
+\sum_\mu|{\cal M}_{H\mu}|^2(n_{\widetilde d,\mu}+1)
\right].
\label{eq:active_Nout}
\end{equation}
Here $n_{\widetilde d,\mu}=\langle\widetilde d_\mu^\dagger\widetilde d_\mu\rangle$ is the occupation of the positive-norm partner defined in Eq.~(\ref{eq:partner_frequency_definition}).
Thus Eq.~(\ref{eq:bog_central_ratio}) alone does not determine the positive noise power in an active channel.
This is why the stimulated-Hawking interpretation of a reflectivity should not automatically be identified with a passive beam-splitter model.

Including an uncorrelated asymptotic input with occupation $n_\infty$, the complete normally ordered output is
\begin{eqnarray}
N_{\out}
&=&|{\cal S}_{\rm tot}|^2n_\infty
+|{\cal F}|^2
\left[
\sum_\alpha|{\cal L}_{H\alpha}|^2n_{c,\alpha}
\right.\nonumber\\
&&\left.
+\sum_\mu|{\cal M}_{H\mu}|^2(n_{\widetilde d,\mu}+1)
\right].
\label{eq:complete_output_occupation}
\end{eqnarray}

The commutator relation fixes a signed difference of transfer strengths, whereas Eq.~(\ref{eq:active_Nout}) fixes positive noise only after the input occupations are specified.
Correlated inputs require cross covariances as well.
In the passive single-input case, the horizon part of the symmetrized covariance is
\begin{equation}
{\cal S}_{\out}^{H}=|{\cal L}_{\rm eff}|^2(n_H+1/2),
\label{eq:Sout_def}
\end{equation}
which must be combined with the incident-port vacuum contribution to obtain Eq.~(\ref{eq:complete_passive_covariance}).

\subsection{Environment output}
The unitary realization in Appendix~\ref{app:rational} includes an unobserved environment output.
For a pure environment input, tracing out the exterior gives a complementary channel.
For a thermal or other mixed input it gives only a weakly complementary channel; a Stinespring complement additionally retains a purifier \cite{CarusoGiovannettiHolevo2006}.
This reduced-channel construction describes absorption within the local model and does not specify the subsequent evolution of the hidden sector during evaporation.

\section{Finite-plane energy estimate and feedback stability}
\label{app:stability}

For the matched exterior problem, the strict return bound can be proved without assigning a separate propagation delay to the asymptotic barrier coefficient.
Put $s=-iz$, so that $\RePart s>0$, and let $\psi_s$ be the exterior solution decaying at infinity.
It solves $-\psi_s''+(V+s^2)\psi_s=0$ on $[x_0,\infty)$.
Before division by the boundary value, integration by parts gives
\begin{equation}
-\frac{\psi_s'(x_0)\psi_s(x_0)^*}{s}
=\frac{1}{s}\int_{x_0}^{\infty}(|\psi_s'|^2+V|\psi_s|^2)\,dx
+s\int_{x_0}^{\infty}|\psi_s|^2\,dx.
\label{eq:exterior_undivided_energy}
\end{equation}
For a nonnegative exterior potential and a nontrivial solution, the real part of the right-hand side is strictly positive when $\RePart s>0$.
Hence $\psi_s(x_0)$ cannot vanish, and the impedance
\begin{equation}
Z(s)=-\frac{\psi_s'(x_0)}{s\psi_s(x_0)}
\label{eq:exterior_impedance}
\end{equation}
is well defined.
Taking the real part after division by $|\psi_s(x_0)|^2$ gives
\begin{eqnarray}
\RePart Z(s)
&=&\frac{\RePart s}{|\psi_s(x_0)|^2}
\left[
\frac{1}{|s|^2}\int_{x_0}^{\infty}
(|\psi_s'|^2+V|\psi_s|^2)\,dx
\right.\nonumber\\
&&\left.
+\int_{x_0}^{\infty}|\psi_s|^2\,dx
\right]>0.
\label{eq:exterior_energy_identity}
\end{eqnarray}
The decaying solution can be chosen locally analytic in $s$; its nonvanishing boundary value makes $Z$ analytic in the right half-plane.
The return referred to this plane is
\begin{equation}
\widehat{\cal R}_{\bh,x_0}(is)
=\frac{1-Z(s)}{1+Z(s)},
\qquad
\left|\widehat{\cal R}_{\bh,x_0}(is)\right|<1.
\label{eq:exterior_return_Schur}
\end{equation}
This verifies the strict barrier bound used in Sec.~\ref{sec:composition} for the matched axial Schwarzschild model.
The essential object is the return at the finite plane.
An asymptotic reflection coefficient with its reference phases removed need not separately be a contractive function throughout the upper half-plane.
The energy argument and the causal decomposition of prompt and delayed Green functions address different aspects of the problem \cite{ZhangRen2026}.

For an admissible horizon response,
\begin{equation}
|\widehat{\cal R}_{\bh,x_0}(z){\cal R}_H(z)|<1,
\qquad \ImPart z>0,
\label{eq:Schur_assumption}
\end{equation}
so the feedback denominator has no upper-half-plane zero.
If the loop product is uniformly bounded by $q<1$, one also has $|D^{-1}|\leq(1-q)^{-1}$.
The rational model of Appendix~\ref{app:rational} has such a bound; the thermal completions approach unit reflectivity at zero frequency and require a separate threshold analysis for any uniform estimate.

\section{Local resonance and phase-crossing formulas}
\label{app:poles}
\label{sec:resonances}

\subsection{Common resonance poles}

Both the horizon-normalized coherent echo amplitude and the hidden-noise transfer contain
\begin{equation}
D(\omega)^{-1}.
\end{equation}
Therefore any simple zero $\omega_n$ of $D$ that is not canceled by a numerator is a pole of both ${\cal K}$ and ${\cal L}_{{\rm eff},\alpha}$.
This gives a direct resonance-level version of the consistency relation and complements the source-dependent residue analysis of Ref.~\cite{KaminagaEcho2026}.

Let $\omega_n$ be a simple zero of $D(\omega)$ in a domain where ${\cal R}_{\bh}$, ${\cal T}_{\bh}$, ${\cal R}_H$, and ${\cal L}_{H\alpha}$ are holomorphic and finite.
If
\begin{equation}
{\cal T}_{\bh}(\omega_n){\cal R}_H(\omega_n)\neq0
\end{equation}
and ${\cal L}_{H\alpha}(\omega_n)\neq0$ for at least one hidden channel, then $\omega_n$ is a common simple pole of the horizon-normalized coherent echo amplitude and of the corresponding hidden-noise transfer amplitude.

The residues differ through the local factors ${\cal R}_H$ and ${\cal L}_{H\alpha}$ multiplying the same filter ${\cal F}$.
The propagation phase $e^{-i\omega x_0}$ is already included in ${\cal F}$.
Analytic continuation of the real-axis modulus identity alone would not be sufficient for this statement; the causal analytic hypothesis on ${\cal L}_{H\alpha}$ is essential.
We do not impose positivity or modulus relations directly at complex frequency.
In the passive case, the corresponding real-axis transfer strengths obey Eq.~(\ref{eq:central_ratio}), and a narrow resonance enhances both coherent return and hidden noise with the same denominator.

\subsection{Pole displacement by a response perturbation}

Let a reference inner response ${\cal R}_{H,0}$ define
\begin{equation}
D_0(\omega)=1-Q_0(\omega),
\qquad
Q_0(\omega)
={\cal R}_{\bh}(\omega)
{\cal R}_{H,0}(\omega)e^{-2i\omega x_0}.
\label{eq:D0_Q0}
\end{equation}
Consider
\begin{equation}
{\cal R}_H(\omega;\epsilon)
={\cal R}_{H,0}(\omega)e^{\epsilon h(\omega)}.
\label{eq:RH_deformation}
\end{equation}

Suppose that $Q_0$ and $h$ are holomorphic near $\omega_n^{(0)}$ and that $\omega_n^{(0)}$ is a simple zero of $D_0$.
For sufficiently small $\epsilon$, the continued pole satisfies
\begin{equation}
\omega_n(\epsilon)
=\omega_n^{(0)}
+\epsilon\delta\omega_n+O(\epsilon^2),
\end{equation}
with
\begin{equation}
\delta\omega_n
=-
\frac{h(\omega_n^{(0)})}
{\left.
\dfrac{d}{d\omega}\log Q_0(\omega)
\right|_{\omega=\omega_n^{(0)}}}.
\label{eq:pole_shift}
\end{equation}

The deformed denominator is
\begin{equation}
D(\omega,\epsilon)
=1-Q_0(\omega)e^{\epsilon h(\omega)}.
\end{equation}
Expanding around $(\omega_n^{(0)},0)$ and using $Q_0(\omega_n^{(0)})=1$ gives Eq.~(\ref{eq:pole_shift}).

For a long cavity of length $L=-x_0$, if the propagation phase dominates the logarithmic derivative,
\begin{equation}
\frac{d}{d\omega}\log Q_0(\omega)\simeq 2iL,
\end{equation}
then
\begin{equation}
\delta\omega_n
\simeq
\frac{i}{2L}h(\omega_n^{(0)}).
\label{eq:long_cavity_shift}
\end{equation}
Thus a response perturbation shifts the common transfer pole.
A real-axis spectral maximum can also depend on the frequency variation of the numerator and the input occupation, so its displacement need not equal the pole displacement.
For a physical response family, passivity and upper-half-plane analyticity must be checked separately for the chosen $h$.

\subsection{Real-axis phase crossings}
For the smooth completion and a simple crossing $\Theta_c(\omega_n^{\rm ph})=2\pi n$, put $\eta=\log(\omega_c/\mu)$ for a fixed reference frequency $\mu>0$.
The scaling in Eq.~(\ref{eq:UV_phase_integral}) gives
\begin{equation}
\partial_\eta\phi_c=\phi_c-\omega\tau_H.
\label{eq:UV_phase_cutoff_derivative}
\end{equation}
At a simple phase crossing, implicit differentiation therefore yields
\begin{equation}
\frac{d\omega_n^{\rm ph}}{d\eta}
=-\frac{\phi_c-\omega\tau_H}
{2L+\partial_\omega\arg{\cal R}_{\bh}+\tau_H}
\bigg|_{\omega=\omega_n^{\rm ph}} .
\label{eq:UV_interference_shift}
\end{equation}
The denominator is the derivative of the complete round-trip phase.
The cutoff changes both this phase and the loop modulus; the numerical calculation in Sec.~\ref{sec:UV_numerics} retains both effects.

\section{A rational causal unitary realization}
\label{app:rational}

\subsection{Causal resonant reflectivity and analytic loss port}

For a concrete stable response model, define the reflection amplitude directly by
\begin{equation}
{\cal R}_H(\omega)
=
\rho\frac{2\gamma\Omega}
{\Omega^2+\gamma^2-\omega^2-2i\gamma\omega},
\qquad
0<\rho<1,
\qquad
0<\gamma<\Omega.
\label{eq:stable_RH}
\end{equation}
Its only poles are
\begin{equation}
\omega=\pm\Omega-i\gamma,
\label{eq:stable_RH_poles}
\end{equation}
so it is analytic in the upper half-plane and satisfies the reality condition in Eq.~(\ref{eq:RH_reality}).
On the real axis,
\begin{equation}
|{\cal R}_H(\omega)|^2
=
\rho^2
\frac{4\gamma^2\Omega^2}
{(\Omega^2+\gamma^2-\omega^2)^2
+4\gamma^2\omega^2}
\leq\rho^2.
\label{eq:stable_RH_bound}
\end{equation}
The maximum is attained at $\omega^2=\Omega^2-\gamma^2$.
The corresponding admittance and self-energy are obtained from Eqs.~(\ref{eq:Lambda_from_RH}) and (\ref{eq:Sigma_from_RH}).
Since $|{\cal R}_H(z)|<1$ in the upper half-plane, $1+{\cal R}_H(z)$ has no zero there.
This construction therefore satisfies both the real-axis passivity bound and the dynamical stability condition.

The passive loss amplitude can also be chosen analytically.
Define
\begin{eqnarray}
A_0&=&\Omega^2+\gamma^2,
\nonumber\\
g_0&=&2\rho\gamma\Omega,
\nonumber\\
B_0&=&\sqrt{A_0^2-g_0^2},
\nonumber\\
\gamma_L&=&
\left(
\frac{B_0+\gamma^2-\Omega^2}{2}
\right)^{1/2}.
\label{eq:spectral_factor_parameters}
\end{eqnarray}
For $0<\rho<1$, $B_0>|\Omega^2-\gamma^2|$ and $\gamma_L$ is real and positive.
The function
\begin{equation}
{\cal L}_H(\omega)
=
\frac{B_0-\omega^2-2i\gamma_L\omega}
{A_0-\omega^2-2i\gamma\omega}
\label{eq:analytic_LH}
\end{equation}
has no upper-half-plane pole and obeys
\begin{equation}
|A_0-\omega^2-2i\gamma\omega|^2-g_0^2
=
|B_0-\omega^2-2i\gamma_L\omega|^2
\label{eq:spectral_factor_identity}
\end{equation}
for real $\omega$.
Therefore
\begin{equation}
|{\cal R}_H(\omega)|^2
+|{\cal L}_H(\omega)|^2
=1.
\label{eq:explicit_dilation_identity}
\end{equation}
For this rational example, an environment output can be included explicitly.
Define $d_H(z)=A_0-z^2-2i\gamma z$, $P_L(z)=B_0-z^2-2i\gamma_Lz$, and $P_L^\sharp(z)=B_0-z^2+2i\gamma_Lz$.
The two-port completion is
\begin{equation}
\left(\begin{array}{c}a_+\\c_{\rm out}\end{array}\right)
=S_H(z)\left(\begin{array}{c}a_-\\c\end{array}\right),
\qquad
S_H(z)=\frac{1}{d_H(z)}
\left(\begin{array}{cc}
g_0&P_L(z)\\-P_L^\sharp(z)&g_0
\end{array}\right).
\label{eq:rational_horizon_completion}
\end{equation}
Its entries are proper rational functions with no upper-half-plane pole.
For real frequency, $P_L^\sharp=P_L^*$, and Eq.~(\ref{eq:spectral_factor_identity}) gives
\begin{equation}
S_H(\omega)S_H(\omega)^\dagger=I.
\label{eq:rational_horizon_unitarity}
\end{equation}
Thus this example has a bounded causal unitary scattering completion, including the unobserved output $c_{\rm out}$.
It is an effective realization used to test the general echo composition and is distinct from the ultraviolet-completed Boltzmann model in Eq.~(\ref{eq:UV_flux_reflectivity}).

We use
\begin{equation}
M\Omega=0.35,
\qquad
M\gamma=0.060,
\qquad
\rho=0.58,
\qquad
x_0/M=-30.
\label{eq:numerical_parameters}
\end{equation}

\subsection{Channel checks with the axial barrier}
We use the axial barrier of Sec.~\ref{sec:numerics} and the integration procedure in Appendix~\ref{app:numerical}.

Figure~\ref{fig:coherent} shows the barrier transmission, the stable near-horizon reflectivity, the inner-normalized echo transfer, and the complete asymptotic channel.
The two panels distinguish the common-filter identity from the interference between prompt reflection and echoes required by asymptotic unitarity.

\begin{figure}[htbp]
\centering
\includegraphics[width=0.82\linewidth]{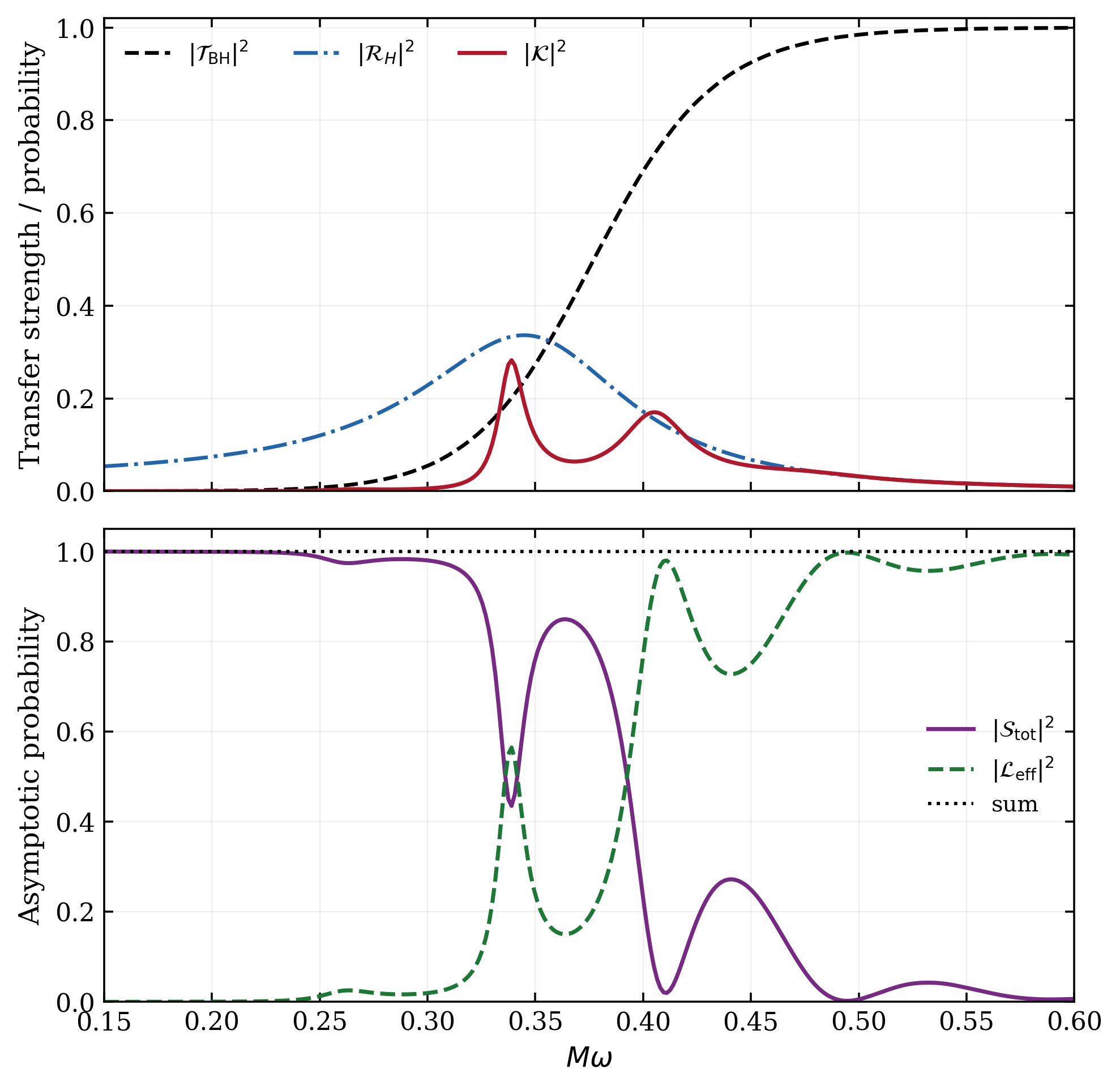}
\caption{Axial $s=2$, $\ell=2$ Schwarzschild example.
The upper panel shows the barrier transmission probability, the stable near-horizon reflection probability generated by Eq.~(\ref{eq:stable_RH}), and the inner-normalized coherent transfer.
The lower panel shows the complete asymptotic reflection $|{\cal S}_{\rm tot}|^2$ and the horizon-port transfer $|{\cal L}_{\rm eff}|^2$; their sum is unity within the numerical accuracy.
The parameters are given in Eq.~(\ref{eq:numerical_parameters}).
The plot was produced using AI-assisted Python code described in Appendix~\ref{app:numerical}.}
\label{fig:coherent}
\end{figure}

Representative values are
\begin{table}[tbp]
\caption{Representative dimensionless scattering quantities for the axial $s=2$, $\ell=2$ example with the stated finite integration boundaries.
The displayed digits specify the numerical setup; they do not imply five-digit accuracy for the infinite-domain problem.}
\label{tab:numerics}
\centering
\begin{tabular}{ccccc}
\toprule
$M\omega$ & $|{\cal R}_{\bh}|^2$ & $|{\cal T}_{\bh}|^2$ & $|{\cal R}_H|^2$ & $|{\cal K}|^2$\\
\midrule
0.30 & 0.94528 & 0.05472 & 0.22830 & 0.00583\\
0.35 & 0.72661 & 0.27339 & 0.33395 & 0.11971\\
0.40 & 0.30978 & 0.69022 & 0.17184 & 0.15950\\
0.45 & 0.07571 & 0.92429 & 0.06780 & 0.05458\\
\bottomrule
\end{tabular}

\end{table}

\begin{table}[tbp]
\caption{Complete passive-channel check for the axial $s=2$, $\ell=2$ example.}
\label{tab:complete_channel}
\centering
\begin{tabular}{cccc}
\toprule
$M\omega$ & $|{\cal S}_{\rm tot}|^2$ & $|{\cal L}_{\rm eff}|^2$ & sum\\
\midrule
0.30 & 0.98030 & 0.01970 & 1.00000000\\
0.35 & 0.76124 & 0.23876 & 1.00000000\\
0.40 & 0.23134 & 0.76866 & 1.00000000\\
0.45 & 0.24952 & 0.75048 & 1.00000001\\
\bottomrule
\end{tabular}

\end{table}

\subsection{Numerical channel checks}

For the single-channel completion in Eq.~(\ref{eq:analytic_LH}), the transfer from the hidden input to infinity is
\begin{equation}
|{\cal L}_{\rm eff}|^2
=
\frac{|{\cal T}_{\bh}|^2
(1-|{\cal R}_H|^2)}{|D|^2}.
\label{eq:numerical_Leff}
\end{equation}
The data used in Fig.~\ref{fig:coherent} satisfy Eq.~(\ref{eq:central_ratio}) pointwise.
We also reconstruct the reflection amplitude from the infinity side using the reciprocity and unitarity phases of Eq.~(\ref{eq:barrier_two_port}).
The resulting complete channel satisfies
\begin{equation}
\left|
|{\cal S}_{\rm tot}|^2+|{\cal L}_{\rm eff}|^2-1
\right|<5\times10^{-8}
\label{eq:numerical_full_unitarity}
\end{equation}
on the frequency grid.
Thus the calculation checks not only the inner transfer ratio but also the interference between prompt reflection and echoes required by the observable asymptotic relation in Eq.~(\ref{eq:Stot_absorption_identity}).
In contrast to a frequencywise square-root choice, Eq.~(\ref{eq:analytic_LH}) fixes a stable causal phase for the loss port.
State dependence then enters through Eq.~(\ref{eq:complete_output_occupation}), after the channel itself has been fixed.

\section{Numerical scattering normalization}
\label{app:numerical}

The Python script \path{generate_bhmap15_figures.py}, used for the numerical calculations and Figs.~\ref{fig:boltzmann_uv}, \ref{fig:UV_quantum_noise}, \ref{fig:UV_phase_cavity}, and \ref{fig:coherent}, was developed using code generated by OpenAI ChatGPT under the author's direction.
The numerical results were checked against the analytical phase formulas, flux-conservation and channel-unitarity identities, and the boundary and grid sensitivity tests described below.
The script is available as ancillary material with the arXiv version of this paper.

With $M=1$, the tortoise coordinate is
\begin{equation}
r_*=r+2\log\left(\frac{r}{2}-1\right).
\label{eq:tortoise_explicit}
\end{equation}
Its inverse can be written using the Lambert function,
\begin{equation}
r(r_*)
=2\left[
1+W\left(e^{r_*/2-1}\right)
\right].
\label{eq:lambert_inverse}
\end{equation}
The outgoing boundary condition at $x_+=140$ is
\begin{equation}
\psi(x_+)=e^{i\omega x_+},
\qquad
\psi'(x_+)=i\omega e^{i\omega x_+}.
\end{equation}
\subsection{Barrier integration and sensitivity checks}

The first-order complex system is integrated with the eighth-order DOP853 method using relative tolerance $2\times10^{-9}$ and absolute tolerance $2\times10^{-11}$.
The scattering spectra use 301 equally spaced frequencies in $0.15\leq M\omega\leq0.60$; Fig.~\ref{fig:boltzmann_uv}, which does not require barrier integration, uses 500 points in $0.002\leq M\omega\leq0.60$.
After inward integration to $x_-=-50$, the near-horizon coefficients are extracted as
\begin{eqnarray}
B_{\rm out}
&=&
\frac12e^{-i\omega x_-}
\left[
\psi(x_-)+\frac{\psi'(x_-)}{i\omega}
\right],
\label{eq:Bout_extract}
\\
B_{\rm in}
&=&
\frac12e^{i\omega x_-}
\left[
\psi(x_-)-\frac{\psi'(x_-)}{i\omega}
\right].
\label{eq:Bin_extract}
\end{eqnarray}
The amplitudes in Eq.~(\ref{eq:BH_RT}) follow directly.
Reciprocity gives ${\cal T}_{\rm in}^{(0)}={\cal T}_{\bh}$ in this reference-plane convention, and barrier unitarity gives
\begin{equation}
{\cal R}_{\infty}
=-
\frac{{\cal T}_{\bh}{\cal R}_{\bh}^*}
{{\cal T}_{\bh}^*}.
\label{eq:Rinfinity_reconstruction}
\end{equation}
This reconstructs the complete asymptotic amplitude in Eq.~(\ref{eq:Stot}).
The values in Tables~\ref{tab:numerics} and \ref{tab:complete_channel}, together with Eqs.~(\ref{eq:numerical_parameters}) and (\ref{eq:Bout_extract})--(\ref{eq:Rinfinity_reconstruction}), specify the normalization and unitarity checks used here.

The plane-wave condition at $x_+$ neglects the exterior potential tail.
Conservation of the Wronskian does not measure this truncation error.
As a boundary-sensitivity check, we repeated the 301-frequency calculation with $x_+=280$ and $560$, keeping the other parameters and tolerances fixed.
Relative to $x_+=140$, the maximum relative changes of each of $|{\cal T}_{\bh}|^2$, $|{\cal K}|^2$, and $|{\cal L}_{\rm eff}|^2$ were below $0.73\%$ and $0.67\%$, respectively.
The figures and tables retain the $x_+=140$ calculation.
These comparisons support the illustrative spectra, but are not certified error bounds for the infinite-domain solution.
The small defects in the flux and channel identities should be interpreted separately.

The finite inner matching approximation is much smaller on this grid.
For $M=1$, $V_{22}(-30)\simeq8.44\times10^{-8}$.
Extracting the free-wave amplitudes at $x_-=-30$, instead of $-50$, changes $|{\cal T}_{\bh}|^2$ by less than $8\times10^{-7}$ relatively and changes $|{\cal K}|^2$ and $|{\cal L}_{\rm eff}|^2$ by less than $1.4\times10^{-6}$.
This is a direct comparison with the matched exterior model used in the finite-plane energy argument.

For Fig.~\ref{fig:UV_quantum_noise}, Eq.~(\ref{eq:outer_phase_quadrature}) is evaluated separately on $[0,\omega]$ and $[\omega,\infty)$ by adaptive quadrature.
At the removable point the integrand is assigned the analytic limit
\begin{equation}
\lim_{\nu\to\omega}
\frac{u_R(\nu)-u_R(\omega)}{\omega^2-\nu^2}
=-\frac{u_R'(\omega)}{2\omega}.
\label{eq:outer_phase_removable_limit}
\end{equation}
The relative and absolute quadrature targets are $2\times10^{-9}$ and $2\times10^{-10}$, respectively.
The same 301-frequency grid and barrier amplitudes are then used in Eq.~(\ref{eq:UV_asymptotic_noise_spectrum}).

For the cutoff comparison, Eq.~(\ref{eq:UV_phase_integral}) is evaluated independently of the subtracted integral in Eq.~(\ref{eq:outer_phase_quadrature}).
On the displayed band the convergent series in Eq.~(\ref{eq:UV_phase_series}) provides a further check.
Phase crossings are found by a bracketed scalar root solve with a cubic-spline interpolation of the continuous barrier phase on the same frequency grid.
Increasing $x_+$ from $140$ to $280$ or $560$ changes the phase crossings in Table~\ref{tab:UV_crossings} by less than $1.5\times10^{-6}$ in $M\omega$.
Increasing the frequency grid from $301$ to $601$ points changes them by less than $10^{-9}$.
The smaller delay-aligned differences require a paired check, since an individual crossing's boundary sensitivity can be comparable to the residual shift.
Recomputing both responses with each outer boundary changes the last-column differences in Table~\ref{tab:UV_crossings} by less than $9.6\times10^{-10}$ in $M\omega$.
Changing the inner extraction point or using 601 frequencies changes those differences by less than $1.6\times10^{-12}$.
Across these checks, the maximum relative delay-aligned cavity difference changes by less than $2.1\times10^{-9}$ in absolute value.
The residual phase is evaluated directly from the convergent difference series with the $j=0$ term removed, avoiding subtraction of the two large phases; the two evaluations agree to within $5\times10^{-15}$ rad at the three verification frequencies.
These are sensitivity checks, not certified error bounds.
The group delays in Table~\ref{tab:UV_phase} use the exact expression in Eq.~(\ref{eq:UV_group_delay}).

\begin{table}[htbp]
\caption{Constructive phase-alignment frequencies and their shifts.
The unaligned shift is $M(\omega_{n,40/M}^{\rm ph}-\omega_{n,20/M}^{\rm ph})$.
The last column gives $10^6 M(\omega_{n,40/M}^{\rm ph}-\omega_{n,20/M}^{\rm ph,(d)})$, where the lower-cutoff response receives the positive delay in Eq.~(\ref{eq:UV_delayed_response}).
The displayed digits describe the finite-domain calculation.}
\label{tab:UV_crossings}
\centering
\begin{tabular}{ccccc}
\toprule
$n$ & $M\omega_c=20$ & $M\omega_c=40$ & unaligned shift & aligned shift $\times10^6$\\
\midrule
4 & 0.21483 & 0.20339 & $-0.01144$ & $1.90$\\
5 & 0.27798 & 0.26343 & $-0.01455$ & $3.84$\\
6 & 0.33917 & 0.32204 & $-0.01713$ & $6.47$\\
\bottomrule
\end{tabular}

\end{table}

\subsection{Matching-scale phase plot}
\label{app:matching_plot}

Figure~\ref{fig:matching_scale} uses $M=1$, $M\omega=0.35$, $M\Omega_{\rm UV}=M\Omega_b=10^4$, and $0.002\leq\zeta\leq0.02$.
Thus both redshifted edges range approximately from $20/M$ to $200/M$, and $\omega/\omega_{c,b}<0.018$ throughout.
At each point, $y$ is obtained from $[y\sqrt{1+y^2}+\operatorname{arcsinh}y]/2=\zeta$, and the exact expressions in Eqs.~(\ref{eq:redshift_y}) and (\ref{eq:x0_y_exact}) determine the cutoff and tortoise position.

The phases are evaluated from Eqs.~(\ref{eq:UV_phase_series}) and (\ref{eq:band_phase_series}).
The plotted smooth residual is computed by separating the $j=0$ term analytically:
\begin{eqnarray}
\Phi_{\rm in}-\Phi_{0,c}
&=&-4M\omega\left[y^2+\log(1+y^2)\right]\nonumber\\
&&+\frac{\beta_H\omega}{\pi}
\sum_{j=1}^{\infty}\left(\frac{\omega}{\omega_c}\right)^{2j}
\left[\frac{\log(\omega_c/\omega)}{2j+1}+\frac{1}{(2j+1)^2}\right].
\label{eq:matching_plot_residual}
\end{eqnarray}
The first line follows exactly from Eqs.~(\ref{eq:redshift_y}) and (\ref{eq:x0_y_exact}).
The exact-band residual uses the same geometric term and the negative series in Eq.~(\ref{eq:band_phase_series}).
The source embeds 81 logarithmically spaced samples.
No barrier solution, fitted delay, or phase alignment is used in this plot.
The comparison therefore tests the near-horizon part of the specified family directly, rather than the complete finite-plane scattering spectrum.

\end{document}